\documentclass[final, journal, onecolumn]{IEEEtran}
\pdfoutput=1
\usepackage{amsthm,amsmath,etex,amssymb,latexsym,comment,etoolbox,cancel,pdfpages,graphicx,tikz,pgfplots,amsbsy,multicol,color,xcolor,adforn,enumerate,framed,relsize, fancyhdr}
\usepackage{caption}
\usepackage{subcaption}
\usepackage{mathtools}
\usepackage{slantsc}
\mathtoolsset{showonlyrefs}
\usepackage[T1]{fontenc}
\usepackage{mathrsfs}
\usepackage{ifthen}
\usepackage{qtxmath}
\usepackage{bbm}
\usepackage{microtype}
\usepackage[noadjust]{cite}
\usepackage[all,cmtip,rotate]{xy}
\usepackage[framemethod=tikz]{mdframed}
\newboolean{islongversion}   
\setboolean{islongversion}{true}   
\usetikzlibrary{positioning,chains,fit,shapes,calc,patterns}
\definecolor{myblue}{RGB}{20,100,160}
\definecolor{mygreen}{RGB}{80,160,80}
\definecolor{myorange}{RGB}{130,80,50}

\ifthenelse{\boolean{islongversion}}
{\usepackage[unicode=true, bookmarks=true,bookmarksnumbered=true,bookmarksopen=true,bookmarksopenlevel=2, breaklinks=false,pdfborder={0 0 0},backref=false,colorlinks=false] {hyperref}
\hypersetup{pdftitle={Plausible Deniability}, pdfauthor={Mayank Bakshi and Vinod Prabhakaran}, pdfpagelayout=TwoColumn, pdfnewwindow=true, pdfstartview=XYZ, plainpages=false}}{}
\usepackage{url}

\allowbreak
\newboolean{isdraft}   
\setboolean{isdraft}{true}   
\ifthenelse{\boolean{islongversion}}{\usepackage[bottom,multiple]{footmisc}}{\usepackage[font=small, labelsep=period, belowskip=-17pt,aboveskip=-3pt]{caption}\setlength{\intextsep}{5pt plus 0pt minus 2pt}\usepackage[bottom,flushmargin,multiple]{footmisc}\makeatletter\def\footnoterule{\relax\kern-5pt\hbox to \columnwidth{\hfill\vrule width \linewidth height 0.2pt\hfill}\kern4.2pt}\makeatother}
\onecolumn
\newlength{\figurewidth}
\makeatletter
\if@twocolumn
\else
\fi
\makeatother

\makeatletter
\if@twocolumn\newcommand{\eqbr}[1][]{\\&\  \ifx\relax#1\relax\else #1\fi}
\else\newcommand{\eqbr}[1][]{}
\fi
\makeatother

\makeatletter
\if@twocolumn\newcommand{\eqbrl}{\Bigg\{}
\else\newcommand{\eqbrl}{}
\fi
\makeatother

\makeatletter
\if@twocolumn\newcommand{\eqbrr}{\Bigg\}}
\else\newcommand{\eqbrr}{}
\fi
\makeatother

\makeatletter
\if@twocolumn\newcommand{\lefteqbr}[2]{\lefteqn{#1}\\& #2}
\else\newcommand{\lefteqbr}[2]{#1 & #2}
\fi
\makeatother

\ifthenelse{\boolean{isdraft}}
	{
		\newcommand{\longonly}[1]{\ifmmode #1\else {\color{red}#1}\fi}
		\newcommand{\shortonly}[1]{\ifmmode #1\else {\color{blue}#1}\fi}
		\newcommand{\shortlong}[2]{\ifmmode #1 #2\else {\color{red}#1} {\color{blue}{#2}}\fi}
	}
	{
		\ifthenelse{\boolean{islongversion}}
			{
				\newcommand{\longonly}[1]{#1}
				\newcommand{\shortonly}[1]{}
				\newcommand{\shortlong}[2]{#2}
			}
			{
				\newcommand{\longonly}[1]{}
				\newcommand{\shortonly}[1]{#1}
				\newcommand{\shortlong}[2]{#1}
			}
	}

\newcommand{\numberthis}{\addtocounter{equation}{1}\tag{\theequation}}
\allowdisplaybreaks

\newcommand{\suppress}[1]{}
\newcommand\xqed[1]{%
  \leavevmode\unskip\penalty9999 \hbox{}\nobreak\hfill
  \quad\hbox{#1}}
\newcommand\tqed{\xqed{$ \triangleleft$}}
\theoremstyle{definition}

\providecommand{\corollaryname}{Corollary}
\providecommand{\theoremname}{Theorem}
\newtheorem{definition}{Definition}

\newtheorem{example}{Example}
\AtEndEnvironment{example}{\tqed\\}

\theoremstyle{definition}
\newtheorem{theorem}{Theorem}
\newtheorem{corollary}{Corollary}
\newtheorem{lemma}{Lemma}
\newtheorem{proposition}{Proposition}
\theoremstyle{remark}
\newtheorem{remark}{Remark}

\DeclareMathOperator*{\argmax}{argmax}
\DeclareSymbolFont{lettersA}{U}{txmia}{m}{it}
 \DeclareMathSymbol{\bbr}{\mathbb}{lettersA}{"92}
 \DeclareMathSymbol{\bbc}{\mathord}{lettersA}{"83}
 \DeclareMathSymbol{\expect}{\mathord}{lettersA}{"85}
 \DeclareMathSymbol{\ent}{\mathord}{lettersA}{"88}
 \DeclareMathSymbol{\MI}{\mathord}{lettersA}{"89}
 \DeclareMathSymbol{\KL}{\mathord}{lettersA}{"84}
\DeclareMathSymbol{\EL}{\mathord}{lettersA}{"8C}
\DeclareMathSymbol{\EM}{\mathord}{lettersA}{"8D}

  \DeclareMathSymbol{\bbf}{\mathord}{lettersA}{"86}
   \DeclareMathSymbol{\bbn}{\mathbb}{lettersA}{"8E}
 \DeclareMathSymbol{\bbg}{\mathord}{lettersA}{"87}
 
\makeatletter
\renewenvironment{proof}[1][\proofname]{\par
  \pushQED{\qed}%
  \normalfont \topsep6\p@\@plus6\p@\relax
  \trivlist
  \item[\hskip\labelsep
        \itshape
    #1\@addpunct{:}]\mbox{}\\*
}{%
  \popQED\endtrivlist\@endpefalse
}
\makeatother

\newcommand{\cA}{\mathscr{A}}
\newcommand{\cB}{\mathscr{B}}

\newcommand{\cC}{\mathscr{C}}

\newcommand{\cG}{\mathcal{G}}

\newcommand{\cK}{\mathscr{K}}
\newcommand{\cM}{\mathscr{M}}

\newcommand{\cR}{\mathscr{R}}

\newcommand{\cS}{\mathscr{S}}
\newcommand{\cT}{\mathscr{T}}
\newcommand{\cU}{\mathscr{U}}
\newcommand{\cV}{\mathscr{V}}
\newcommand{\cW}{\mathscr{W}}
\newcommand{\cX}{\mathscr{X}}
\newcommand{\cY}{\mathscr{Y}}
\newcommand{\cZ}{\mathscr{Z}}

\newcommand{\bB}{\mathbf{{B}}}

\newcommand{\bk}{k}
\newcommand{\bK}{K}
\newcommand{\bkA}{k_A}
\newcommand{\bKA}{K_A}
\newcommand{\bkB}{k_B}
\newcommand{\bKB}{K_B}
\newcommand{\bkC}{k_C}
\newcommand{\bKC}{K_C}
\newcommand{\bm}{m}
\newcommand{\bM}{M}
\newcommand{\bmh}{\hat{m}}
\newcommand{\bMh}{\hat{M}}

\newcommand{\bs}{s}
\newcommand{\bS}{S}
\newcommand{\bt}{t}
\newcommand{\bT}{T}
\newcommand{\bu}{{\mathbf{u}}}
\newcommand{\bU}{\mathbf{U}}

\newcommand{\bv}{\mathbf{v}}
\newcommand{\bV}{\mathbf{V}}
\newcommand{\bw}{\mathbf{w}}
\newcommand{\bW}{\mathbf{W}}
\newcommand{\bbW}{\mathbf{W}}

\newcommand{\bx}{\mathbf{x}}
\newcommand{\bX}{\mathbf{X}}

\newcommand{\by}{\mathbf{y}}
\newcommand{\bY}{\mathbf{Y}}

\newcommand{\bz}{\mathbf{z}}
\newcommand{\bZ}{\mathbf{Z}}

\newcommand{\Prob}{{\normalfont\textsf{{P}}}}
\newcommand{\Probc}{{\normalfont\textsf{{Q}}}}
\newcommand{\share}{T}

\newcommand{\REQ}{\mathscr{R}_{\tiny\mbox{\normalfont
equiv}}}
\newcommand{\RBCC}{\mathscr{R}_{\tiny\mbox{\normalfont
bcc}}}

\newcommand{\Func}[1]{\mbox{\scshape #1}}
\newcommand{\Enc}{\Func{Enc}}
\newcommand{\Dec}{\Func{Dec}}
\newcommand{\Fake}{\Func{Fake}}
\newcommand{\Msg}{\Func{Msg}}
\newcommand{\Fak}[1]{{#1}^{{\mbox{\tiny\sc (f)}}}}
\newcommand{\converseeps}{\gamma}
\newcommand{\AEN}[1]{{\cA_{\epsilon}^{(n)}(#1)}}
\newcommand{\set}[1]{\left\{#1\right\}}

\begin{document}
\bstctlcite{IEEEexample:BSTcontrol}
\title{Plausible Deniability over Broadcast Channels}
\author{Mayank~Bakshi\IEEEauthorrefmark{1}~\IEEEmembership{Member,~IEEE,} and Vinod~Prabhakaran\IEEEauthorrefmark{2}~\IEEEmembership{Member,~IEEE}
%\IEEEauthorblockA{\IEEEauthorrefmark{1}Institute of Network Coding, The Chinese University of Hong Kong, Hong Kong. \emph{mayank@inc.cuhk.edu.hk}}
%\IEEEauthorblockA{\IEEEauthorrefmark{2}Tata Institute of Fundamental Research, India. \emph{vinodmp@tifr.res.in}}
\thanks{\IEEEauthorrefmark{1}Mayank Bakshi (\emph{mayank@inc.cuhk.edu.hk}) is with the Institute of Network Coding, The Chinese University of Hong Kong. The work described in this paper was partially supported by a grant from University Grants Committee of the Hong Kong Special Administrative Region, China (Project No. AoE/E-02/08).}
\thanks{\IEEEauthorrefmark{2}Vinod Prabhakaran (\emph{vinodmp@tifr.res.in}) is with the Tata Institute of Fundamental Research, India. Vinod Prabhakaran's research was funded in part by a Ramanujan fellowship from the Department of Science and Technology, Government of India and in part by Information Technology Research Academy (ITRA), Government of India under ITRA-Mobile grant ITRA/15(64)/Mobile/USEAADWN/01.}
\thanks{A preliminary version of this work was presented at the 2016 IEEE International Symposium on Information Theory, Barcelona, Spain.}}
\maketitle
\begin{abstract}
In this paper, we introduce the notion of Plausible Deniability in an information theoretic framework.  We consider a scenario where an entity that eavesdrops through a broadcast channel summons one of the parties in a communication protocol to reveal their message (or signal vector). It is desirable that the summoned party have enough freedom to produce a fake output that is likely plausible given the eavesdropper's observation.  We examine three variants of this problem -- Message Deniability, Transmitter Deniability, and Receiver Deniability. In the first setting, the message sender is summoned to produce the sent message. Similarly, in the second and third settings, the transmitter and the receiver are required to produce the transmitted codeword, and the received vector respectively. For each of these settings, we examine the maximum communication rate that allows a given minimum rate of plausible fake outputs. For the Message and Transmitter Deniability problems, we fully characterise the capacity region for general broadcast channels, while for the Receiver Deniability problem, we give an achievable rate region for physically degraded broadcast channels. 
\end{abstract}
\section{Introduction}
The explosive growth in information technologies in recent years is not without its pitfalls. On one hand, advances in communications have enabled ground-breaking applications that have arguably  been instrumental in improving the general quality of life. On the other hand, the naturally connected nature of these technologies also presents a wide variety of security and privacy concerns. To counter these, much recent attention has also focused on designing and analyzing algorithms and protocols that guarantee security or privacy. It is worth noting that the security requirement often varies greatly with the application. Indeed, the consequences of security failure as well as the nature of eavesdropping parties differ from application to application. For example, for a user posting on a social network, the implication is often limited to loss of personal information to a potentially malicious party. On the other hand, for an whistleblower posting sensitive information to an accomplice, any security failure has potentially life-altering consequences. The nature of the eavesdropper is also different in these situations. In the first example, an eavesdropper is typically a passive party that simply listens to an ongoing transmission, and it is desirable that the content of the communication be kept hidden from the eavesdropper. On the other hand, in the second example, the eavesdropper may often be an authority that has the power to \emph{coerce} the whistleblower to reveal the transmitted message. In this case, it is important that the whistleblower is able to \emph{deny} the fact that any sensitive communication has taken place by producing a fake message that appears plausible to the coercing party.%\footnote{Discerning readers may place our discussion here in the context of ongoing debate around ``metadata'' and its implications for privacy~\cite{Landau:13}. Our notion of Plausible Deniability is partly motivated by this debate.}

We argue that while much of the work in secure communication is well suited to the first scenario, {\em i.e.}, the ability to \emph{hide} data, there is relatively little work that applies to the second scenario. For the first scenario, by now, there is are well developed theoretical results as well as practical algorithms both in the \emph{cryptographic}~\cite{KatzL:07} as well as \emph{information theoretic}~\cite{Wyner:75,CsiszarK:78,BlochB:11} settings. However, there is limited understanding of both fundamental limits and algorithms for the second setting. In this paper, we propose an information theoretic framework for \emph{Plausibly Deniable} communication in the sense just described. In the following, we begin with an overview of some related notions of security and contrast these with our notion of {Plausibly Deniable communication}. 
  
\subsection{Related notions}
\subsubsection{Information theoretic secrecy}
\label{sec:infosecrecy}
Usually secure protocols aim to hide data from an eavesdropper by taking advantage of  some \emph{asymmetry} between the legitimate receiver and the eavesdropper -- the eavesdropper should be ``\emph{less powerful}'' than the legitimate receiver. The framework of  information theoretic secrecy relies on the eavesdropper having ``less information'' than the intended receiver and provides guarantees that hold irrespective of the eavesdropper's computational ability. For example, in the wiretap channel setting~\cite{Wyner:75,CsiszarK:78} (See Figure~\ref{fig:WTC}) the eavesdropper may observe Alice's transmission through a noisier channel than Bob does. On similar lines, in the secure network coding setting~\cite{CaiY:02}, the eavesdropper may observe a smaller subset of the transmission than legitimate nodes. In each of these settings, the information-theoretic approach allows characterizing the ``capacity'', which is defined as the maximum code rate such that \emph{(a)} the intended receiver can decode the secret message $\bm$ {\em reliably} given her received vector $\by$, {\em i.e.}, $\Prob(\hat{\bm}(\by)\neq\bm)\approx 0$, and \emph{(b)} the eavesdropper can gain very little statistical information about the secret message $\bm$ given her observation $\bz$, \emph{i.e.}, $\Prob(\bm|\bz)\approx\Prob(\bm)$. Note here that there is \emph{no} restriction placed on the computational power of the eavesdropper. As a result, schemes that guarantee information theoretic security are free of computational assumptions and as a result are guaranteed to be secure against any future developments in fast computing.  

We argue that even though information theoretic secrecy is perfectly suited when the goal is to only hide the data against a passive eavesdropper, it does not guarantee any protection against eavesdroppers that have the ability to summon one of the communicating parties. The reason for this is as follows. At a high level, information theoretic secrecy is achieved by ensuring that the eavesdropper has a large enough list of candidate messages that appear roughly equiprobable. On the other hand, plausible deniability requires the summoned party to produce one such candidate message \emph{without knowing} the eavesdropper's channel realisation. The following example illustrates this difference more concretely.
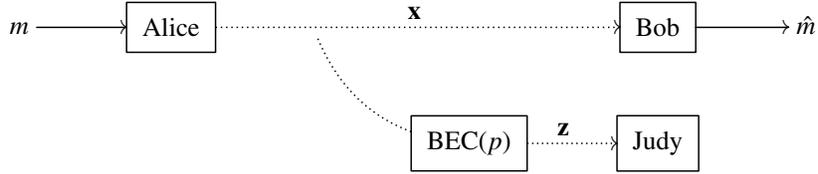
\begin{figure}[t]\begin{displaymath}
\renewcommand{\labelstyle}{\textstyle}
\scalebox{1}{
\xymatrix@C+10pt{
\bm\ar@{->}[r]&
 *++[F]{\mbox{Alice}}\ar@{.>}[rrr]^{ \bx}&\ar@{.}@/_1pc/[dr]&
 &*++[F]{\mbox{Bob}}\ar[r]&\bmh \\
  &&&*++[F]{\mbox{BEC($p$)}}\ar@{.>}[r]^{ \bz}&*++[F]{\mbox{Judy}}\\
% &\bkA\in\cK\ar[u]& \bx\mbox{ or }\Fak{\bx}=\Fake(\bx,\bkA)\ar@{->}@/_1pc/[dr]&\bz\ar[d]\\
%&&&*++[F]{\mbox{Judy}}}
}}
\end{displaymath}\caption{Alice wishes to communicate a message $\bm$ to Bob by sending a codeword $\bx$ over a noiseless binary channel while an eavesdropper Judy observes $\bx$ through a binary erasure channel with erasure probability $p>0$. Note that, in order to avoid being detected as lying, the summoned party's output should appear plausible to Judy given her side information $\bz$. In particular, for the channel in this example, both Alice and Bob are forced to reveal their true codewords ({\em i.e.}, $\bx$) to Judy.  This example also shows a contrast between the standard notion of secrecy and the plausible deniability requirement.}\label{fig:examplebec}
\end{figure}\begin{example}[Secrecy does not guarantee plausible deniability]\label{eg:pdmotivation}  Consider the setting of Figure~\ref{fig:examplebec}. Since the channel to Bob is noiseless, the secrecy capacity~\cite{CsiszarK:78} is $p$. On the other hand, even if Alice and Bob operate a code equipped with an information-theoretic secrecy guarantee and Judy {\em demands} that Alice provide the transmitted codeword $\bx$, Alice has no choice but to provide exactly what was transmitted (and hence, also reveal the message). If Alice chooses to provide a vector $\bx'$ different from $\bx$, then Judy would be able to detect with a constant probability that Alice is lying since the transmitted symbol for any coordinate where $\bx'$ and $\bx$ differ would be received correctly by Judy with probability $1-p$.\end{example}
\subsubsection{Cryptographic security}\label{sec:crypto}
In the cryptographic setting, the asymmetry between the legitimate receiver and the eavesdropepr usually manifests itself through complexity theoretic notions. For example, in a \emph{public key cryptosystem}, the receiver holds a pair of carefully chosen keys $(\bk_{\mbox{\small public}},\bk_{\mbox{\small private}})$. The public key $\bk_{\mbox{\small public}}$ is known to all parties including the eavesdropper, while the private key $\bk_{\mbox{\small private}}$  is known only to the eavesdropper. This allows the sender to encrypt the message $\bm$ to the ciphertext $\bx=\Enc(\bm,\bk_{\mbox{\small public}})$. The encryption algorithm is chosen such that the receiver can use his private key to decrypt the ciphertext to obtain the message as $\bm=\Dec(\bx,\bk_{\mbox{\small public}},\bk_{\mbox{\small private}})$ in polynomial time. On the other hand, without knowing $\bk_{\mbox{\small private}}$, the eavesdropper cannot efficiently compute $\Enc^{-1}(\bx,\bk_{\mbox{\small public}})$ (under reasonable computational assumptions). However, even if the eavesdropper is unable to invert the ciphertext on their own, if they have the ability to summon the receiver to produce the private key, the receiver may have no choice but to respond truthfully by revealing the true private key, else the ciphertext and the public key may not be consistent with it.

%Even though the cryptographic approach provides security guarantees that appear practically reasonable, it is not entirely without drawbacks. Firstly, security guarantees that hinge on computational assumptions may not be information theoretically pleasing -- in theory, given enough computational power, an eavesdropper can invert the ciphertext.\footnote{DON'T NEED THIS. A case in point is the popular RSA scheme~\cite{RivestSA:83} that relies on the assumption that factorization is a ``hard'' problem using classical models of computation; however, it is known that using quantum computers, factorization can be accomplished in polynomial time~\cite{Shor:94}, and hence, the RSA encryption is not secure if the eavesdropper gets access to a quantum computer!} Secondly, even if the eavesdropper is unable to invert the ciphertext on their own, if they have the ability to summon the receiver to produce the private key, the receiver may have no choice but to respond truthfully by revealing the true private key, else the ciphertext and the public key may not be consistent with it.

\subsubsection{Deniable Encryption}\label{sec:DE}
The notion of Deniable Encryption was first introduced by Canetti \emph{et al.} in~\cite{CanettiDNO:97} recognizing  the above problem of lack of plausible deniability in the cryptographic setting.\footnote{Also related is the notion of \emph{uncoercible communication} introduced by Benaloh et al.~\cite{BenT:94}.} Here, the typical setting is as follows. Consider a public key setting as described in Section~\ref{sec:crypto}. Unlike the setting of Section~\ref{sec:crypto} the eavedropper Judy who has bounded computational power both observes the ciphertext and can issue a summon to Bob coercing him to revealing the message. The framework of Deniable Encryption allows for encryption schemes such that upon receiving Judy's summon, Bob is able to produce a fake private key $\Fak{\bk}_{\mbox{\small public}}$ which decrypts the ciphertext to a fake message $\Fak{\bm}$ while appearing plausible to Judy. In other words, there is no polynomial time algorithm, using which Judy is able to determine whether Bob has responded with the true public key or a fake public key. Note that usual public key protocols such as RSA do not allow Bob to produce a fake key for every pair of $(\bm,\bk_{\mbox{\small public}})$. This notion has received much attention in recent years. By now, there are fairly extensive theoretical and practical developments along this line (\emph{c.f.}~\cite{ONeillPW:11,Truecrypt,SahaiW:14} and the references therein).

\subsubsection{Covert Communication}\label{sec:coco}
In both the secrecy and the plausible deniability problems considered above, while the goal is to be able to hide the message that is being transmitted, the implicit assumption is that it is permissible for some form of communication to take place. However, in the setting of \emph{covert commmunication}~\cite{BashGT:13,CheBJ:13,Bloch:16,WangWZ:16}, even the fact that any communication is taking place is objectionable from the eavesdropper's point of view. For example, the communicating parties may be two  prisoners in adjacent cells that wish to communicate without the warden knowing that they are doing so. In this setting, the goal is to ensure that from the warden's point of view, the output distribution induced by non-zero transmissions appear close to that under zero transmission. The capacity for this problem is now well understood and follows the so called \emph{square-root law} -- in $n$ channel uses, only $O(\sqrt{n})$ message bits can possibly be transmitted without being detected. Note that the notion of covertness only guarantees that the eavesdropper be unable to distinguish no transmission from a non-zero transmission; it does not necessarily prevent the eavesdropper from gaining any information about the potential message, if she assumes that something was transmitted.\footnote{One can also demand both covertness and secrecy simultaneously. By operating at even lower rates (though still $O(\sqrt{n})$ bits per $n$ channel uses), it is possible to be covert about the transmission status and secret about the message being potentially transmitted.~\cite{CheBCJ:14,Bloch:16}.}  Therefore, the covertness requirement only implies a weak form of plausible deniability -- the transmitter can claim that no transmission took place when something was transmitted. However, it does not necessarily allow the communicating parties to claim the transmission of a message different from the true message. 
\subsection{Our work} \label{sec:itpd}
Taking inspiration from the formulation of Deniable Encryption discussed in Section~\ref{sec:DE},  we propose an information theoretic approach to plausible deniability. While the approach in Section~\ref{sec:DE} relies on cryptographic assumptions, \emph{i.e.}, the assumption that the eavesdropper is computationally limited without access to the receiver's private key, we assume that the eavesdropper has potentially unlimited computational power, but the eavesdropper and the legitimate receiver have different channels statistics. In this setting, the sender can leverage this difference by careful encoding that allows the receiver to decode the message correctly while leaving enough room for confusion such that, if summoned, transmitter and the receiver are able produce fake messages or codewords that appear statistically indistinguishable from the true message or codeword to the eavesdropper given his channel observation. 
\subsubsection{Our setup} Our general setup is as follows. Alice, Bob, and Charlie are three participants in a potentially secretive communication setup. Charlie  wishes to send a message $\bm\in\cM$ to Bob through Alice. Alice and Bob are at two ends of a noisy channel and operate the physical layer with Alice being the transmitter and Bob being the receiver, while Charlie interacts directly with  Alice and knows the message but does not  partake in the physical layer transmission and reception. The nature of the message may either be an innocuous  or a secretive one -- this is known to Alice, Bob, and Charlie, but not to any eavesdroppers. 

Judy is an eavesdropper who observes a noisy version of Alice's transmission. In this work,  we assume that the statistics of Judy's observation are known to the above three parties, but the exact observation is unknown. We consider three settings for this problem. In the Transmitter Deniability problem, Judy may summon Alice and ask her to produce the transmitted codeword. Similarly, in the Receiver Deniability, and the Message Deniability problems, Judy may  summon Bob, and Charlie, to produce the received vector, and the message, respectively. In each of these settings, depending on whether the communication is innocuous or secretive, the summoned party may either respond truthfully or use a {\em Faking Procedure} to produce a fake output that reveals as little information about the true message as possible while still maintaining plausibility with respect to Judy's observation. 

We quantify the efficacy of a communication scheme in terms of its two properties -- the {\em reliability} of the code  and the {\em plausible deniability} of the faking procedure. The first property {\em i.e.}, the reliability is measured in a  standard fashion in terms of the {\em message rate} and the {\em error probability} at the decoder. Plausible deniability is also measured in terms of two metrics -- the {\em plausibility} and the {\em rate of deniability}. Roughly speaking, plausibility measures the closeness between two distributions -- the joint distribution of the fake output with the eavesdropper's observation and that of the true  message or signal vector with the eavesdropper's observation. We measure this distance in terms of the Kullback-Leibler (K-L) divergence.\footnote{Although, in this paper, we measure the plausibility in terms of K-L divergence, one is also well justified to instead use other measures of distance such as the variational distance. We argue that K-L divergence is a stronger measure for our problem as requiring that the K-L divergence be small also implies that the variational distance is small (by invoking Pinsker's inequality). Further, using K-L divergence instead of variational distance considerably simplifies our converse proofs. It is worth noting that  the variational distance has a natural interpretation in terms of Hypothesis Testing -- the  variational distance between two probability measures $\Prob_1$ and $\Prob_2$ equals $1-\Pr(\textrm{test outputs }\Prob_2|\textrm{true distribution is }\Prob_1)-\Pr(\textrm{test outputs }\Prob_1|\textrm{true distribution is }\Prob_2)$ for an optimal hypothesis test for distinguishing $\Prob_1$ and $\Prob_2$. }  The rate of deniability is measured as the conditional entropy of the fake message given the summoned party's observations. This attempts to capture the amount of freedom the summoned party has while responding to the summons. The rate of deniability may also be roughly interpreted as a measure of equivocation at the eavesdropper after the summoned party is forced to respond. Strictly speaking, the rate of deniability is a purely operational characteristic of the faking procedure and our formal definition of the rate of deniability does not appear to be related to equivocation. However, when the faking procedure satisfies the plausibility requirement, we establish an asymptotic equivalence between these two notions in Propositions~\ref{prop:msgDEQ} and~\ref{prop:cweq}. We also emphasise here that demanding a rate of deniability $D$ is a stronger requirement than demanding an equivocation $D$ in the usual information theoretic secrecy setting -- this naturally extends similar observations in the cryptographic setting where, a plausibly deniable protocol trivially also satisfies the security requirement.

\subsubsection{Organization of this paper}The rest of this paper is organised as follows. In Section~\ref{sec:formulation}, we formally describe our notation and problem formulation and state the main results in Section~\ref{sec:mainresults}. In Sections~\ref{sec:message} and~\ref{sec:codeword}, we give proof sketches for our theorems, and discuss some examples and key properties of our capacity regions. Finally, in Section~\ref{sec:discussions}, we provide concluding remarks.

\section{Problem Formulation}\label{sec:formulation}

\begin{figure*}\begin{center}
\begin{subfigure}{\textwidth}\begin{displaymath}
\renewcommand{\labelstyle}{\textstyle}
 \scalebox{1}{
\xymatrix@C+10pt{
&\bkC\in\cK\ar[d]&\bkA\in\cK\ar[d]&&&*++[F]{\mbox{Bob}}\ar[r] &\bmh=\Dec(\by)\\
{\bm\in\cM}\ar[r]&*++[F]{\mbox{Charlie}}\ar@{->}[r]^{\bm}\ar@/_2pc/[drrrr]_(.5){\Fak{\bm}={\textrm{\sc Fake}}(\bm,\bkC)}&
 *++[F]{\mbox{Alice}}\ar@{.>}[rr]^{ \bx={\textrm{\sc Enc}}(\bm,\bkA)}&&*++[F]{p(y,z|x)} \ar@{.>}@/^1pc/[ur]^{\by} \ar@{.>}@/_1pc/[dr]^{\bz}&
 & &\\
  &&&
&&*++[F]{\mbox{Judy}}\\
% &\bkA\in\cK\ar[u]& \bx\mbox{ or }\Fak{\bx}=\Fake(\bx,\bkA)\ar@{->}@/_1pc/[dr]&\bz\ar[d]\\
%&&&*++[F]{\mbox{Judy}}}
}}
\end{displaymath}\caption{Message deniability}\label{fig:messagedeniability}
\end{subfigure}\\*[4em]
\begin{subfigure}{\textwidth}\begin{displaymath}
\renewcommand{\labelstyle}{\textstyle}
 \scalebox{1}{
\xymatrix@C+10pt{
&&\bkA\in\cK\ar[d]&&&*++[F]{\mbox{Bob}}\ar[r] &\bmh=\Dec(\by)\\
{\bm\in\cM}\ar[r]&*++[F]{\mbox{Charlie}}\ar@{->}[r]^{\bm}&
 *++[F]{\mbox{Alice}}\ar@{.>}[rr]^{ \bx={\textrm{\sc Enc}}(\bm,\bkA)}\ar@/_2pc/[drrr]_(.5){\Fak{\bx}={\textrm{\sc Fake}}(\bx,\bkA)}&&*++[F]{p(y,z|x)} \ar@{.>}@/^1pc/[ur]^{\by} \ar@{.>}@/_1pc/[dr]^{\bz}&
 & &\\
  &&&
&&*++[F]{\mbox{Judy}}\\
% &\bkA\in\cK\ar[u]& \bx\mbox{ or }\Fak{\bx}=\Fake(\bx,\bkA)\ar@{->}@/_1pc/[dr]&\bz\ar[d]\\
%&&&*++[F]{\mbox{Judy}}}
}}
\end{displaymath}\caption{Transmitter Deniability}\label{fig:transmitterdeniability}
\end{subfigure}\\*[4em]
\begin{subfigure}{\textwidth}\begin{displaymath}
\renewcommand{\labelstyle}{\textstyle}
 \scalebox{1}{
\xymatrix@C+10pt{
&&&&&\bkB\in\cK\ar[d]\\
&&\bkA\in\cK\ar[d]&&&*++[F]{\mbox{Bob}}\ar[r]\ar@/^2pc/[dd]^(.5){\Fak{\by}={\textrm{\sc Fake}}(\by,\bkB)} &\bmh=\Dec(\by)\\
{\bm\in\cM}\ar[r]&*++[F]{\mbox{Charlie}}\ar@{->}[r]^{\bm}&
 *++[F]{\mbox{Alice}}\ar@{.>}[rr]^{ \bx={\textrm{\sc Enc}}(\bm,\bkA)}&&*++[F]{p(y,z|x)} \ar@{.>}@/^1pc/[ur]^{\by} \ar@{.>}@/_1pc/[dr]^{\bz}&
 & &\\
  &&&&&*++[F]{\mbox{Judy}}\\
% &\bkA\in\cK\ar[u]& \bx\mbox{ or }\Fak{\bx}=\Fake(\bx,\bkA)\ar@{->}@/_1pc/[dr]&\bz\ar[d]\\
%&&&*++[F]{\mbox{Judy}}}
}}
\end{displaymath}\caption{Receiver Deniability}\label{fig:receiverdeniability}
\end{subfigure}\end{center}
\caption{The above figure shows the three different problem settings considered in this paper. These settings have the following commonalities: Charlie knows only the message $\bm$ and may have access to an independently generated private random string $\bkC$; Alice knows the message $\bm$, an independently generated  private random string $\bkA$ and the transmitted codeword $\bx$; Bob observes the channel output $\by$ and potentially has an independently generated  private random string $\bkB$, and is required to reconstruct $\bm$;  Judy observes the channel output $\bz$. However, depending on the setting we consider, Judy summons Charlie, Alice, or Bob to produce $\bm$, $\bx$, or $\by$ respectively. The summoned party responds with a fake output $\Fake(\cdot)$ that has roughly the same distribution as the variable Judy demands to know. In each setting, the fake output is a function of the true value of variable demanded and the independent private randomness available to the summoned party.}\label{fig:setup}
\end{figure*}
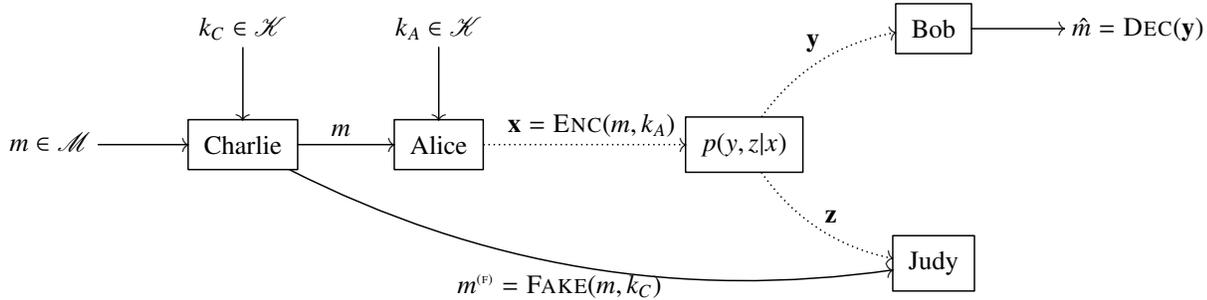
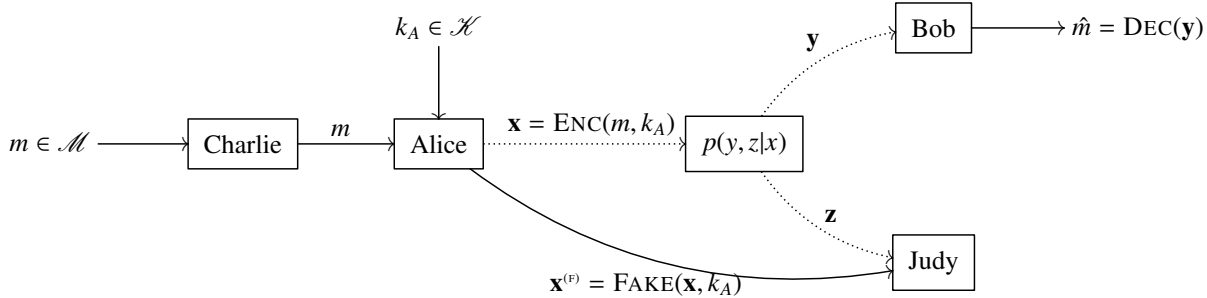
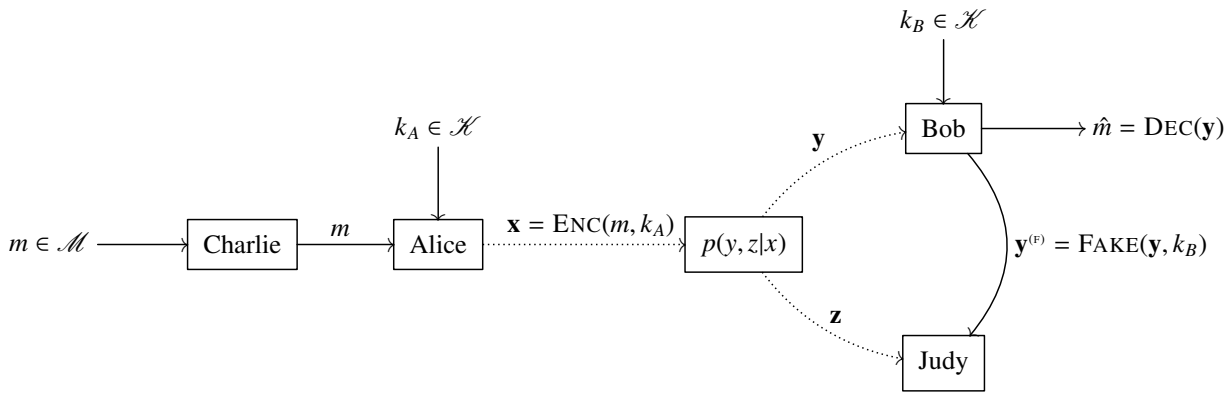

\subsection{Notation} Throughout this paper, we typically adopt the following notation. Upper case math and lower case symbols such as $X$ and $x$ denote random variables and their specific values respectively. Boldface symbols such as $\bX$ and $\bx$ denote random vectors and their specific values respectively, while calligraphic symbols such as $\cX$ denote sets. Probability distributions of generic random variables is typically written as $\Prob$ (e.g. $\Prob_X$, $\Prob_{Y|X}$), while  probability distributions imposed by the specific codebook are typically written as $\Probc$ (e.g. $\Probc_{\bX}$). All logarithms in this paper are assumed to base $2$. For some random variables $X$ and $Y$ following distributions $\Prob_X$ and $\Prob_Y$ on alphabets $\cX$ and $\cY$ respectively, we define the entropy, conditional entropy, and the mutual information respectively as $\ent(X)\triangleq \sum_{x\in\cX:\Prob_X(x)>0}\Prob_X(x)\log\left(1/\Prob_X(x)\right)$, $\ent(Y|X)\triangleq \sum_{x\in\cX:\Prob_{X,Y}(x,y)>0}\Prob_{X,Y}(x,y)\log\left(1/\Prob_{Y|X}(y|x)\right)$,  and $\MI(X;Y)=\ent(X)-\ent(X|Y)$. The Kullback-Leibler divergence between two probability measures $\Prob_1$ and $\Prob_2$ over a set $\cX$ is defined as $\KL(\Prob_1||\Prob_2)\triangleq\sum_{x\in\cX:\Prob_1(x)>0}\Prob_1(x)\log\left(\Prob_1(x)/\Prob_2(x)\right)$. Throughout this paper, we employ strong typicality in our analysis, and define the strongly typical set for a random variable $X$ as $$\AEN{X}\triangleq\left\{\bx\in\cX^n: \max_{x\in\cX}\left|\frac{|\{i:x_i=x\}|}{n}-\Prob_X(x)\right|\leq\frac{\epsilon}{|\cX|}\right\}.$$
\subsection{Channel model} Consider the problem settings shown in Figure~\ref{fig:setup}. Alice, Bob, and Judy are connected through  the following memoryless broadcast channel -- at each discrete time instant, Alice's transmission $X\in\cX$, Bob's reception $Y\in\cY$,  and Judy's observation $Z\in\cZ$ follow the conditional distribution $\Prob_{Y,Z|X}$ over finite alphabets $\cX\times\cY\times\cZ$. Initially, only Charlie knows the message $\bm\in\cM$ and passes it onto Alice to be transmitted to Bob over the broadcast channel. Charlie only knows the value of the message, but does not see the channel inputs or outputs. Throughout this paper, we assume that the message $\bM$ is uniformly distributed over $
\cM$. There is no shared randomness, but Alice, Bob, and Charlie have private randomness $\bKA\in\cK$, $\bKB\in\cK$, and $\bKC\in\cK$ respectively. In addition, the code and the faking procedure (defined in the following) are known to all parties.

\subsection{Codes and Faking Procedures} A code of block-length $n$ is a pair of maps $\Enc:\cM\times\cK\to\cX^n$ and $\Dec:\cY^n\to\cM$. These maps are applied by Alice and Bob to generate the codeword $\bx\triangleq x^n=\Enc(\bm,\bkA)$ and the reconstruction $\bmh=\Dec(\by)$ respectively.  When there is no private randomness at Alice, we denote the codeword for message $\bm$ by $\bx(\bm)$. To simplify notation, we represent a code $(\Enc,\Dec)$ through its codebook $\cC\triangleq\{\Enc(\bm,\bkA):\bm\in\cM,\bkA\in\cK\}$. Note that $\cC$ is a multi-set with possible repetitions as we do not require that $\Enc(\cdot)$ be an injective map.  

Judy may summon Alice, Bob, or Charlie to provide a variable $\bw\in\cW$ that can be used to reconstruct the message using a map $\Msg:\cW\to\cM$. Depending on whether or not the transmission is an innocuous, the summoned party may either reveal the true value of $\bw$ or use a (possibly stochastic) {\em faking procedure}  $\Fake:\cW\times\cK\to\cW$ to output a fake  value $\Fak{\bw}\in\cW$.
In this paper, we consider three settings that are specified by the choice of the variable $\bw$. In particular, we consider the following special cases:
\paragraph{Message deniability} This setting is shown in Figure~\ref{fig:messagedeniability}. Charlie is the summoned party, $\bw=\bm$, $\cW=\cM$, and $\Msg(\bw)=\bw$.
\paragraph{Transmitter deniability}  This setting is shown in Figure~\ref{fig:transmitterdeniability}. Here, Alice is the summoned party, $\bw=\bx$, $\cW=\cX^n$, and $\Msg(\bw)$ is the most likely message given that $\bx=\bw$, {\em i.e.}, $\Msg(\bw)\triangleq\argmax_{\bm\in\cM}\Probc_{\bM|\bX}(\bm|\bw)$ if the maximum is attained at a unique value of $\bm$. If there are multiple values of $\bm$ achieving the above maximum, then $\Msg(\bw)$ selects one of them arbitrarily.
\paragraph{Receiver deniability} This setting is shown in Figure~\ref{fig:receiverdeniability}. Bob is the summoned party, $\bw=\by$, $\cW=\cY^n$, and $\Msg(\bw)=\Dec(\bw)$.

\subsection{Reliability} We say that $\cC$ is  $(\epsilon,R)$-reliable if $\frac{1}{n}\log|\cM|=R$, and there exists an encoder and decoder pair $(\Enc,\Dec)$ such that the average error probability $\sum_{(\bm,\by):{\small\Dec}(\by)\neq\bm}\Probc_{\bM,\bY}(\bm,\by)$  is no larger than $\epsilon$.  Here, $\Probc_{\bM}$ is the uniform distribution on $\cM$ and $\Probc_{\bY,\bM}$ is the joint distribution of the message $\bM$ and Bob's received vector $\bY$ that induced by the specific code $(\Enc,\Dec)$ and the channel transition probability $\Prob_{YZ|X}$. 

\subsection{Plausible deniability} We first define our notion of plausible deniability for general random variables, and subsequently, specialise it to our setting. Let $\Fak{\bbW}$, $\bbW$, and $\bZ$ be random variables distributed according to a distribution $\Probc_{\Fak{\bbW},\bbW,\bZ}$. Let  $\Probc_{\bZ,\bbW}$ and $\Probc_{\bZ,\Fak{\bbW}}$ be marginals of the distribution $\Probc_{\bZ,\bbW,\Fak{\bbW}}$.  We say that $\Fak{\bbW}$ is $(\delta,D)$-\emph{plausibly deniable} for $\bbW$ given observation $\bZ$ if 
\begin{enumerate}[(i)]
\item $\KL(\Probc_{\bZ,\Fak{\bbW}}||\Probc_{\bZ,\bbW})\leq\delta$, and 
\item $\frac{1}{n}\ent(\Msg(\Fak{\bbW})|\bbW)= D$.
\end{enumerate}

In this paper, we are interested in settings where $\bbW$ is the random variable whose value is demanded by Judy through her summon, $\Fak{\bbW}$ is the random variable denoting the output of the faking procedure $\Fake(\cdot)$  employed by the summoned party, and $\bZ$ is Judy's observation. The parameters $\delta$ and $D$ respectively measure the \emph{plausibility} and the \emph{rate of deniability} of $\Fake(\cdot)$. We say that a faking procedure $\Fake(\cdot)$ is $(\delta,D)$-plausibly deniable for $\bbW$ given observation $\bZ$ is its output $\Fak{\bbW}$ is $(\delta,D)$-plausibly deniable for $\bbW$ given observation $\bZ$.
\begin{remark} Note that since we assume that the output of the faking procedure depends only value of variable $\bbW$ (that is known to the summoned party) and the summoned party's independently distributed private randomness, the random variables $\Fak{\bbW}$, $\bbW$, $\bZ$ satisfy the Markov chain $\Fak{\bbW}-\bbW-\bZ$. \end{remark}\begin{remark} Note that the joint distribution  $\Probc_{\bZ,\bbW,\Fak{\bbW}}$  depends on both the code $(\Enc,\Dec)$ and the faking procedure, $\Fake(\cdot)$ and takes into account the (uniform) message distribution $\Probc_\bM$, the channel conditional probability $\Prob_{YZ|X}$, and the distribution of independent private randomness variables $\bKA$, $\bKB$, and $\bKC$.\end{remark}

\subsection{Capacity regions}
For each setting $\bw\in\{\bm,\bx,\by\}$, we say that a rate-deniability pair $(R,D)$ is achievable if for any $\epsilon,\delta>0$, for some $R'\geq R$ and $D'\geq D$,  and for large enough $n$, there exists a blocklength-$n$ code $\cC$ that is $(\epsilon,R')$-reliable and a faking procedure $\Fake(\cdot)$ that is $(\delta,D')$-plausibly deniable for $\bW$ given $\bZ$. The capacity region $\cR_{\bw}$ is the closure of the set of all achievable rate-deniability pairs.

\section{Main Results}\label{sec:mainresults}
For the message deniability problem, we give a characterisation the capacity region $\cR_{\bm}$ for general broadcast channels in Theorem~\ref{thm:message}. The proof of this theorem is presented in Section~\ref{sec:message}.
\begin{theorem}[Message Deniability]\label{thm:message}\normalfont$\cR_{\bm}$ is the set of all $(R,D)$ pairs such that \begin{align*}
&0\leq R\leq \MI(Y;V)+\MI(U;Y|V)-\MI(U;Z|V),\mbox{ and}\\
&0\leq D\leq \min\set{R,\MI(U;Y|V)-\MI(U;Z|V)}
\end{align*}
for some random variables  $U$ and $V$ which take values in sets $\cU$ and $\cV$, respectively, with $|\cU|\leq \left(|\cX|+1\right)\left(|\cX|+2\right)$ and $|\cV|\leq |\cX|+2$, and  satisfy the Markov chain $V - U - X- (Y,Z)$. 
\end{theorem}
Next, we characterise the capacity region $\cR_{\bx}$ for the transmitter deniability problem for general broadcast channels and given an achievable region for the receiver deniability problem for physically degraded broadcast channels. These results are stated in Theorems~\ref{thm:transmitter} and~\ref{thm:rxachievability} below and are proved in Section~\ref{sec:codeword}.
\begin{theorem}[Transmitter Deniability]\label{thm:transmitter} \normalfont$\cR_{\bx}$ is the set of all $(R,D)$ pairs  such that  
\begin{align*}
&0\leq R\leq \MI(X;Y),\mbox{ and} \\
&0\leq D\leq \min\set{R,\MI(X;Y|U)}
\end{align*}
for some random variable $U$ which takes values in a set $\cU$, with $|\cU|\leq|\cX|$, and satisfes the Markov chains $U-X-(Y,Z)$ and $X-U-Z$.
\end{theorem}
\begin{theorem}[Achievability for Receiver Deniability]\label{thm:rxachievability} \normalfont Let $\Prob_{Y,Z|X}$ be a physically degraded broadcast channel, {\em i.e.}, $\Prob_{Z|X}(z|x)=\sum_{y\in\cY}\Prob_{Z|Y}(z|y)\Prob_{Y|X}(y|x)$ for some distribution $\Prob_{Z|Y}$. Then, $\cR_{\by}$ includes all $(R,D)$ pairs such that 
\begin{align*}
&0\leq R\leq \MI(X;Y),\mbox{ and}\\
&0\leq D\leq \min\set{R,\MI(X;Y|V)}
	\end{align*}
	for some random variable $V$ which takes values in a finite set $ \cV$  and satisfies the Markov chains $V - Y- (X,Z)$ and $Y-V-Z$.
\end{theorem}

\section{Message Deniability}
\label{sec:message}
In this section, we outline the proof of Theorem~\ref{thm:message} and discuss connections of the message deniability problem with standard information theoretic  secrecy problems. Our achievability argument relies on reducing our problem to the following variant of the information theoretic secrecy problem.
\subsection{Broadcast channel with confidential and leaked messages}
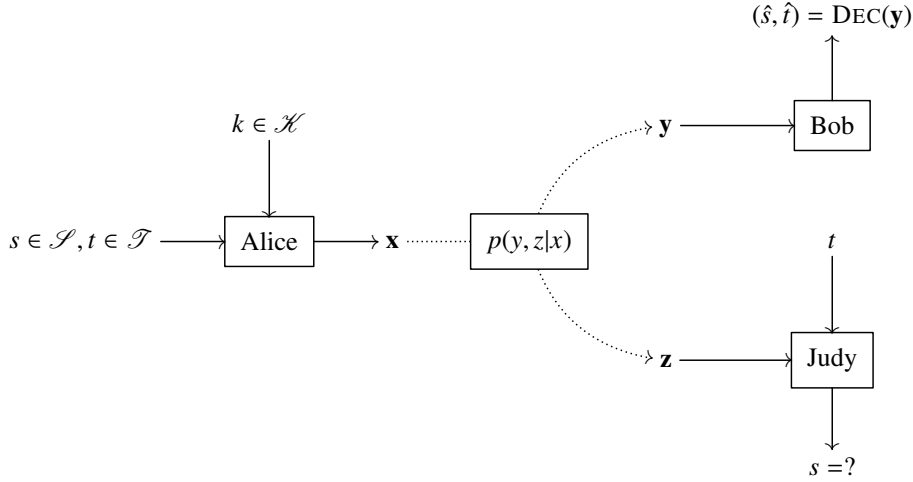
\begin{figure}
\begin{displaymath}
 \scalebox{1}{
\xymatrix{
&&&&& (\hat{\bs},\hat{\bt})=\Dec(\by)\\
&\bk\in\cK\ar[d] &&&\by \ar[r]&
 *++[F]{\mbox{Bob}}\ar[u] &
 \\
\bs\in\cS,\bt\in\cT\ar[r] &
 *++[F]{\mbox{Alice}}\ar[r]& \bx  \ar@{.}[r]&*++[F]{p(y,z|x)}\ar@/^1pc/@{.>}[ru]\ar@/_1pc/@{.>}[rd] &
 & \bt\ar[d] \\
 &&&&\bz\ar[r]&*++[F]{\mbox{Judy}}\ar@{->}[d]\\
&&&&&\bs=?}}\end{displaymath}\caption{Any code for the above secrecy problem can be operated as a code for the Message Deniability problem by treating $\bs$ as the part of the message that the faking algorithm randomizes over and $\bt$ as the part of the message that is unchanged by it.}\label{fig:secrecy}
\end{figure}
 Consider the setup shown in Figure~\ref{fig:secrecy}. Alice observes sources $\bs\in\cS$ and $\bt\in\cT$ and wishes to transmit them reliably to Bob over $n$ uses of the channel. Judy observes a noisy version of the transmission and knows the source $\bt$ as side information. The goal for the transmission is to ensure that the leakage $\MI(\bS;\bZ|\bT)$ is small. At first sight, the setting here is similar to the public message and confidential message setting of \cite{CsiszarK:78} in that secrecy is only required for the private message $\bs$. However, in contrast to~\cite{CsiszarK:78}, Judy is not interested in estimating $\bt$ based on $\bz$, but is instead provided with $\bt$ as side-information. This allows us to operate at potentially higher rates than~\cite{CsiszarK:78}. We define the capacity region for this problem in the following.
 \begin{definition}\label{def:bccsi} The capacity region $\cR_{s}$ for broadcast channel with confidential and side-information messages is the set of $(R_\bs,R_\bt)$ pairs such that, given $\epsilon,\delta>0$, a large enough blocklength $n$, and sources $\bS$ and $\bT$ drawn independently and uniformly from $\cS$ and $\cT$ respectively, there exists a code $\cC$, consisting of an encoder $\Enc: \cS\times\cT\times\cK\to\cX^n$, a decoder $\Dec:\cY^n\to\cS\times\cT$, and Alice's private randomness $\bK\in\cK$, that satisfies the following properties:
\begin{enumerate} 
	\item $|\cS|\geq 2^{nR_\bs}$ and $|\cT|\geq 2^{nR_\bt}$. 
	\item \label{cond:bccsierror}$\Probc_{\bS,\bT,\bY}(\Dec(\bY)\neq(\bS,\bT))\leq \epsilon$.
	\item \label{cond:strongsecrey} $\MI\left(\bS;\bT,\bZ\right)<\delta$.
\end{enumerate}
\end{definition}
 The following lemma provides an inner bound on $\cR_{s}$. 
 \begin{lemma}\label{lem:sideinfo} $\cR_{s}$ includes the set of all $(R_{\bs},R_{\bt})$ pairs such that there exist random variables $U$ and $V$ satisfying $V - U - X - (Y,Z)$, 
 \begin{align}
 	 R_{\bs}&\leq \MI(U;Y|V) - \MI(U;Z|V),\mbox{  and}\label{eq:boundrs}\\
 	 R_{\bt}& \leq \MI(V;Y).\label{eq:boundrt}
 \end{align}  
 \end{lemma}
 The above lemma gives an achievable region for this problem with strong secrecy (condition~\ref{cond:strongsecrey} of Definition~\ref{def:bccsi}). In the following corollary, we show that for every rate pair in this region, there exists a code for which the K-L divergence between the distributions $\Probc_{\bS}\Probc_{\bT,\bZ}$  and $\Probc_{\bS,\bT,\bZ}$ is small. This property is useful in the proof of Theorem~\ref{thm:message}, where we show that codes for the above secrecy problem lead to suitable codes and faking procedure for our message deniability problem. 
 \begin{corollary}\label{cor:sideinfo} Lemma~\ref{lem:sideinfo} continues to hold if the condition  $\KL(\Probc_{\bS}\Probc_{\bT,\bZ}||\Probc_{\bS,\bT,\bZ}) <\delta$ is added to Definition~\ref{def:bccsi} .
\end{corollary}
 We discuss the proof of Lemma~\ref{lem:sideinfo} and Corollary~\ref{cor:sideinfo} in  Appendix~\ref{app:strongsecrecysideinfo}.
\subsection{Proof of achievability in Theorem~\ref{thm:message}}
 It suffices to prove the achievability of $(R,D)$ pairs satisfying $V-U-X-(Y,Z)$, 
 \begin{align*}
&0\leq R\leq \MI(Y;V)+D,\mbox{ and}\\
&0\leq D\leq \MI(U;Y|V)-\MI(U;Z|V).
\end{align*}
Note that such an $(R,D)$ pair may be expressed as $R=R_{\bs}+R_{\bt}$, and $D=R_{\bt}$, where the pair $(R_{\bs},R_{\bt})$ satisfies the inequalities~\eqref{eq:boundrs} and~\eqref{eq:boundrt} specified in Lemma~\ref{lem:sideinfo}. The crux of the achievability proof is the following reduction argument. Let $\epsilon,\delta>0$ be given. Choose $n$ large enough so that there exists a code $\cC$ of rate $(R_{\bs},R_{\bt})$ satisfying the achievability of Corollary~\ref{cor:sideinfo} with the chosen values of $\epsilon$ and $\delta$. For the message deniability problem, we decompose the $nR$-length message $\bm$ into two parts --  a confidential part $\bs$  of $nR_{\bs}$ bits, and a leaked part $\bt$ of $nR_{\bt}$ bits. Next, Alice and Bob encode and decode $(\bs,\bt)$ using the code $\cC=(\Enc,\Dec)$. The reliability guarantees for our code thus follow directly from the guarantees on $\cC$  proved in Corollary~\ref{cor:sideinfo}. The faking procedure draws $\bs'$ independently at random from the distribution $\Probc_{\bS}$ on $\{0,1\}^{nR_{\bs}}$ and outputs $\Fak{\bm}=(\bs',\bt)$. For the faking procedure thus constructed, 
\begin{align}
\KL(\Probc_{\Fak{\bM},\bZ}||\Probc_{\bM,\bZ})&=	\KL(\Probc_{\bS',\bT,\bZ}||\Probc_{\bS,\bT,\bZ})\\
&= \KL(\Probc_{\bS}\Probc_{\bT,\bZ}||\Probc_{\bS,\bT,\bZ})\\
&\overset{(a)}{\leq} \delta. 
\end{align}
In the above, the bound (a) follows from the guarantees provided in Corollary~\ref{cor:sideinfo}.  This shows that $(R,D)\in\cR_{\bm}$.
\subsection{Proof of converse in Theorem~\ref{thm:message}}
The scheme described in the previous section has the following property. Given the  part of the message that is revealed to Judy, the additional information learnt by Judy based on her channel observation is no larger than $\delta$. In particular this implies that for the scheme presented in our achievability proof, $\MI(\bM;\bZ|\Fak{\bM})<\delta$, \emph{i.e.}, given the fake message, the channel observation and the true message are nearly independent. In our converse proof, we start off by showing that this property must, in fact, be true for any faking procedure that satisfies the plausibility requirement.  Further, we also show that in order for a faking procedure to be plausible,  the entropy for the message and the fake message must be close each other. The following lemma makes these claims precise.
 \begin{lemma} \label{lem:convmessage}Let $\Fak{\bM}$ be $(\delta,D)$-plausibly deniable for $\bM$ given observation $\bZ$ and satisfy $\Fak{\bM}-\bM-\bZ$. Then, there exists a non-negative constant $\lambda$ depending only on $\Prob_{Z|X}$ and $|\cM|$ such that
\begin{align}
& \MI(\bM;\bZ|\Fak{\bM}) \leq \delta+n\lambda\sqrt{\delta},\mbox{ and}\label{eq:mzgivenmf}\\
& \lvert \ent(\bM) - \ent(\Fak{\bM})\rvert \leq \delta+n\lambda\sqrt{\delta}\label{eq:hmhmf}.
\end{align}
\end{lemma}
\begin{proof}
We explicitly  prove only the first inequality. The second inequality follow from a similar reasoning. We first use the definition of mutual information and Kullback-Leibler Divergence to note that 
\begin{align*}
\lefteqn{\MI(\bM;\bZ|\Fak{\bM})=\ent(\bZ|\Fak{\bM})-\ent(\bZ|\bM)}\\
=&{\mathlarger\sum_{(\bz,\bm):\Probc_{\bZ,\Fak{\bM}}(\bz,\bm)>0}}\Probc_{\bZ,\Fak{\bM}}(\bz,\bm)\log{\frac{\Probc_{\Fak{\bM}}(\bm)}{\Probc_{\bZ,\Fak{\bM}}(\bz,\bm)}} \eqbr -{\mathlarger\sum_{(\bz,\bm):\Probc_{\bZ,\bM}(\bz,\bm)>0}}\Probc_{\bZ,\bM}(\bz,\bm)\log{\frac{\Probc_{\bM}(\bm)}{\Probc_{\bZ,\bM}(\bz,\bm)}}\\
=&{\mathlarger\sum_{\bm:\Probc_{\Fak{\bM}}(\bm)>0}}\Probc_{\Fak{\bM}}(\bm)\log{\Probc_{\Fak{\bM}}(\bm)} -{\mathlarger\sum_{\bm\in\cM}}\Probc_{\bM}(\bm)\log{\Probc_{\bM}(\bm)}\\
&\ -{\mathlarger\sum_{(\bz,\bm):\Probc_{\bZ,\Fak{\bM}}(\bz,\bm)>0}}\Probc_{\bZ,\Fak{\bM}}(\bz,\bm)\log{\frac{1}{\Probc_{\bZ,\Fak{\bM}}(\bz,\bm)}} \eqbr +{\mathlarger\sum_{(\bz,\bm):\Probc_{\bZ,\bM}(\bz,\bm)>0}}\Probc_{\bZ,\bM}(\bz,\bm)\log{\frac{1}{\Probc_{\bZ,\bM}(\bz,\bm)}}\\
= &\ {D}\left(\Probc_{\Fak{\bM}}||\Probc_{\bM}\right) +{\mathlarger\sum_{\bm\in\cM}}\left[\Probc_{\Fak{\bM}}(\bm)-\Probc_{\bM}(\bm)\right]\log\Probc_{\bM}(\bm)\\
&\  - {D}\left(\Probc_{\bZ,\Fak{\bM}}||\Probc_{\bZ,\bM}\right)  - \eqbr{\mathlarger\sum_{(\bz,\bm):\Probc_{\bZ,\Fak{\bM}}(\bz,\bm)>0}}\Probc_{\bZ,\Fak{\bM}}(\bz,\bm)\log{\frac{1}{\Probc_{\bZ,\bM}(\bz,\bm)}}\eqbr -{\mathlarger\sum_{(\bz,\bm):\Probc_{\bZ,\bM}(\bz,\bm)>0}}\Probc_{\bZ,\bM}(\bz,\bm)\log{\frac{1}{\Probc_{\bZ,\bM}(\bz,\bm)}} \\
= &\ {D}\left(\Probc_{\Fak{\bM}}||\Probc_{\bM}\right) +{\mathlarger\sum_{\bm\in\cM}}\left[\Probc_{\Fak{\bM}}(\bm)-\Probc_{\bM}(\bm)\right]\log\Probc_{\bM}(\bm)\\
&\  - {D}\left(\Probc_{\bZ,\Fak{\bM}}||\Probc_{\bZ,\bM}\right) \eqbr - {\mathlarger\sum_{(\bz,\bm):\Probc_{\bZ,\bM}(\bz,\bm)>0}}\left[\Probc_{\bZ,\Fak{\bM}}(\bz,\bm)-\Probc_{\bZ,\bM}(\bz,\bm)\right]\log{\frac{1}{\Probc_{\bZ,\bM}(\bz,\bm)}}.\eqbr\numberthis\label{eq:mzgivenf1}
\end{align*}
In the last step we use the fact that $\{(\bz,\bm):\Probc_{\bZ,\Fak{\bM}}(\bz,\bm)>0\}\subseteq \{(\bz,\bm):\Probc_{\bZ,\bM}(\bz,\bm)>0\}$ as $\KL(\Probc_{\bZ,\Fak{\bM}}||\Probc_{\bZ,\bM})<\delta<\infty$. Continuing further from Eq.~\eqref{eq:mzgivenf1} and again using the fact that $\KL(\Probc_{\bZ,\Fak{\bM}}||\Probc_{\bZ,\bM})<\delta$, we have
\begin{align*}
\lefteqbr{\MI(\bM;\bZ|\Fak{\bM})}{\overset{(a)}{\leq}\ \delta - {\mathlarger\sum_{(\bz,\bm):\Probc_{\bZ,\bM}(\bz,\bm)>0}}\left[\Probc_{\bZ,\Fak{\bM}}(\bz,\bm)-\Probc_{\bZ,\bM}(\bz,\bm)\right]\eqbr[\qquad\times] \log{\frac{1}{\Probc_{\bZ,\bM}(\bz,\bm)}}}\\
\leq &\ \delta + {\mathlarger\sum_{(\bz,\bm):\Probc_{\bZ,\bM}(\bz,\bm)>0}} \eqbrl \left|\Probc_{\bZ,\Fak{\bM}}(\bz,\bm)-\Probc_{\bZ,\bM}(\bz,\bm)\right| \eqbr[\qquad\times] \log{\frac{1}{\Probc_{\bZ,\bM}(\bz,\bm)}} \eqbrr\\
= &\ \delta + {\mathlarger\sum_{(\bz,\bm):\Probc_{\bZ,\bM}(\bz,\bm)>0}} \eqbrl \left|\Probc_{\bZ,\Fak{\bM}}(\bz,\bm)-\Probc_{\bZ,\bM}(\bz,\bm)\right| \eqbr[\qquad\times] \log{\frac{|\cM|}{\sum_{\bx:\Prob_{\bZ|\bX}(\bz|\bx)>0}\Prob_{\bZ|\bX}(\bz|\bx)\Probc_{\bX|\bM}(\bx|\bm)}} \eqbrr\\
\overset{(b)}{\leq} &\ \delta +{\mathlarger\sum_{(\bz,\bm):\Probc_{\bZ,\bM}(\bz,\bm)>0}} \eqbrl \left|\Probc_{\bZ,\Fak{\bM}}(\bz,\bm)-\Probc_{\bZ,\bM}(\bz,\bm)\right| \eqbr[\qquad\times]  \mathlarger{\sum}_{\substack{\bx\in\cX^n\\\Prob_{\bZ|\bX}(\bz|\bx)>0}}\Probc_{\bX|\bM}(\bx|\bm)\log{\frac{|\cM|}{\Prob_{\bZ|\bX}(\bz|\bx)}} \eqbrr\\
\leq &\ \delta +{\mathlarger\sum_{(\bz,\bm):\Probc_{\bZ,\bM}(\bz,\bm)>0}} \eqbrl \left|\Probc_{\bZ,\Fak{\bM}}(\bz,\bm)-\Probc_{\bZ,\bM}(\bz,\bm)\right| \eqbr[\qquad\times]  \mathlarger{\sum}_{\substack{\bx\in\cX^n\\ \Prob_{\bZ|\bX}(\bz|\bx)>0}}\Probc_{\bX|\bM}(\bx|\bm)\max_{\bx':\Prob_{\bZ|\bX}(\bz|\bx')>0}\log{\frac{|\cM|}{\Prob_{\bZ|\bX}(\bz|\bx')}} \eqbrr \\
\leq &\ \delta +{\mathlarger\sum_{(\bz,\bm):\Probc_{\bZ,\bM}(\bz,\bm)>0}}\eqbrl \left|\Probc_{\bZ,\Fak{\bM}}(\bz,\bm)-\Probc_{\bZ,\bM}(\bz,\bm)\right| \eqbr[\qquad\times]  \mathlarger{\sum}_{\substack{\bx\in\cX^n\\ \Prob_{\bZ|\bX}(\bz|\bx)>0}}\Probc_{\bX|\bM}(\bx|\bm)n\log\frac{(|\cM|)^{1/n}}{\min_{(z,x):\Prob_{Z|X}(z|x)>0}\Prob_{Z|X}(z|x)}\eqbrr\\
\overset{(c)}{\leq} &\ \delta+n\sqrt{2\delta}\left[\log{|\cM|}-\log{\frac{1}{\min_{(z,x):\Prob_{Z|X}(z|x)>0}\Prob_{Z|X}(z|x)}}\right].
\end{align*}
In the above, $(a)$ follows by using the fact that $\Fak{\bM}$ is $(\delta,D)$-plausibly deniable for $\bM$ given $\bZ$ to bound the first term in~\eqref{eq:mzgivenf1}, noting that $\Probc_\bM(\bm)$ equals $1/|\cM|$ to conclude that the second term is zero, and applying the non-negativity of the Kullback-Leibler divergence. The inequality $(b)$ is obtained by using Jensen's inequality. Finally, $(c)$ follows applying Pinsker's inequality to bound the variational distance between the distributions $\Probc_{\bZ,\Fak{\bM}}$ and $\Probc_{\bZ,\bM}$.\end{proof}

\begin{proof}[Proof of converse of Theorem~\ref{thm:message}] Let $\epsilon,\delta>0$. We begin by obtaining $n$-letter bounds on $D$ and $R$ for any $(\epsilon,R)$-reliable and $(\delta,D)$-plausibly deniable code. To this end, from the definition and Lemma~\ref{lem:convmessage}, there exists $\converseeps=\gamma(\epsilon,\delta)>0$ such that $\lim_{(\epsilon,\delta)\to (0,0)}\gamma=0$, and
\begin{align}
nD &\leq  \ent(\Fak{\bM}|\bM)  \notag\\
&=  \ent(\bM|\Fak{\bM}) + \ent(\Fak{\bM}) - \ent(\bM) \notag\\
&\leq \ent(\bM|\Fak{\bM}) + n\converseeps\notag\\
 &\overset{(a)}{\leq} \MI(\bM;\bY|\Fak{\bM}) + 2n\converseeps\notag\\
 &\leq \MI(\bM;\bY|\Fak{\bM}) - \MI(\bM;\bZ|\Fak{\bM}) + 3n\converseeps. \label{eq:converse_generic_D}\\
\intertext{In the above, $(a)$ follows by applying Fano's inequality and letting $\gamma$ be at least as large as $\epsilon$. Next, Applying we apply Fano's inequality to bound the rate $R$ as}
nR &\leq \MI(\bM;\bY) + n\converseeps\\
 & = \MI(\Fak{\bM},\bM;\bY) + n\converseeps\notag\\
 &= \MI(\Fak{\bM};\bY) + \MI(\bM;\bY|\Fak{\bM}) + n\converseeps\notag\\
 &\leq \MI(\Fak{\bM};\bY) + \MI(\bM;\bY|\Fak{\bM}) - \MI(\bM;\bZ|\Fak{\bM}) + 2n\converseeps, \label{eq:converse_generic_R}
\end{align}
where the second equality follows from the fact that $\Fak{\bM}-\bM-\bY$ is a Markov chain. Next, we obtain single-letter versions of the above expressions. Let $\share$ be uniformly distributed over $[1:n]$ and independent of $(\bM,\Fak{\bM},\bX,\bY,\bZ)$. From \eqref{eq:converse_generic_D},
\begin{align*}
\lefteqn{D \leq \frac{1}{n}\left[\MI(\bM;\bY|\Fak{\bM}) - \MI(\bM;\bZ|\Fak{\bM})\right]+3\gamma }\\
 &\overset{(a)}{=}\frac{1}{n} 
\sum_{i=1}^n \left[\MI(\bM;Y_i|Y^{i-1},Z_{i+1}^n,\Fak{\bM}) - \MI(\bM;Z_i|Y^{i-1},Z_{i+1}^n,\Fak{\bM})\right]+3\gamma \\
 &= \MI(\bM;Y_\share|Y^{\share-1},Z_{\share+1}^n, \Fak{\bM}, \share) - \MI(\bM;Z_\share|Y^{\share-1},Z_{\share+1}^n,\Fak{\bM},\share) +3\gamma
 \end{align*}
 where $(a)$ follows from Csisz\'ar's sum identity~\cite{ElGamalK:11}. Also,
 \begin{align*}
\MI(\Fak{\bM};\bY) &=\ \sum_{i=1}^n \MI(\Fak{\bM};Y_i|Y^{i-1}) \\
& \leq \sum_{i=1}^n \MI(\Fak{\bM},Y^{i-1},Z_{i+1}^n;Y_i)\\
 &= n\MI(\Fak{\bM},Y^{\share-1},Z_{\share+1}^n;Y_\share|\share)\\
 & \leq n\MI(\Fak{\bM},Y^{\share-1},Z_{\share+1}^n,\share;Y_\share).
 \end{align*}
Hence, from~\eqref{eq:converse_generic_R}, 
\begin{align*}
R&\leq  \MI(\Fak{\bM},Y^{\share-1},Z_{\share+1}^n,\share;Y_\share)+\MI(\bM;Y_\share|Y^{\share-1},Z_{\share+1}^n, \Fak{\bM}, \share) \eqbr -\MI(\bM;Z_\share|Y^{\share-1},Z_{\share+1}^n,\Fak{\bM}) + 2\converseeps.
\end{align*} 
Next, let $V=(\Fak{\bM},Y^{\share-1},Z_{\share+1}^n,\share)$, $U=(V,\bM)$, $X=X_\share$, $Y=Y_\share$ and $Z=Z_\share$. Then, clearly, $V - U - X - (Y,Z)$. Substituting above and letting $\epsilon$ and $\delta$ be arbitrarily small (but positive)  shows that  any achievable rate-deniability pair $(R,D)$ must satisfy
\begin{align*}
&0\leq R\leq \MI(Y;V)+\MI(U;Y|V)-\MI(U;Z|V),\mbox{ and}\\
&0\leq D\leq \min\set{R,\MI(U;Y|V)-\MI(U;Z|V)}
\end{align*}
for some random variables $U$ and $V$  satisfying the Markov chain $V - U - X - (Y,Z)$. 

Finally,  we argue that it suffices to consider random variables $U$ and $V$ such that  $|\cU|\leq \left(|\cX|+1\right)\left(|\cX|+2\right)$ and $|\cV|\leq |\cX|+2$. The proof  follows along the cardinality bounding  argument for the broadcast channel with confidential messages~\cite[pp.~347-348]{CsiszarK:78}. In particular, consider  auxiliary variables $V$ and $U$, that take values in sets $\cV$ and $\cU$ respectively, and are jointly distributed with $X,Y$, and $Z$ such that $\Prob_{VUXYZ}(v,u,x,y,z)=\Prob_V(v)\Prob_{U|V}(u|v)\Prob_{X|U}(x|u)\Prob_{YZ|X}(y,z|x)$ for every $(v,u,x,y,z)\in\cV\times\cU\times\cX\times\cY\times\cZ$. The first step in the proof is to show that there exist auxiliary variables $\tilde{V}$ and $\tilde{U}$, that take values in sets $\tilde{\cV}$ and $\cU$ respectively, are jointly distributed with $X,Y$, and $Z$ such that  $\Prob_{\tilde{V}\tilde{U}XYZ}(v,u,x,y,z)=\Prob_{\tilde{V}}(v)\Prob_{U|V}(u|v)\Prob_{X|U}(x|u)\Prob_{YZ|X}(y,z|x)$ for every $(v,u,x,y,z)\in\tilde{\cV}\times\cU\times\cX\times\cY\times\cZ$, where, $\Prob_{\tilde{V}}$ satisfy the following constraints:
\begin{align}
&\sum_{v\in\tilde{\cV}} \Prob_{\tilde{V}}(v)\sum_{u\in\cU}\Prob_{U|V}(u|v)\Prob_{X|U}(x|u)=	\sum_{v\in{\cV}} \Prob_{{V}}(v)\sum_{u\in\cU}\Prob_{U|V}(u|v)\Prob_{X|U}(x|u)=\Prob_X(x)\mbox{ for all $x\in\cX$},
\label{eq:card1}\\
&\sum_{v\in\tilde{\cV}} \Prob_{\tilde{V}}(v)\ent(Y|{V}=v)=\sum_{v\in\tilde{\cV}} \Prob_{{V}}(v)\ent(Y|V=v),\label{eq:card2}\\ 
&\sum_{v\in\tilde{\cV}} \Prob_{\tilde{V}}(v)\left(\ent(Y|V=v)-\sum_{u\in\cU}\Prob_{U|V}(u|v)\ent(Y|U=u)\right)=\sum_{v\in{\cV}} \Prob_{{V}}(v)\left(\ent(Y|V=v)-\sum_{u\in\cU}\Prob_{U|V}(u|v)\ent(Y|U=u)\right),\label{eq:card3}\\
&\sum_{v\in\tilde{\cV}} \Prob_{\tilde{V}}(v)\left(\ent(Z|V=v)-\sum_{u\in\cU}\Prob_{U|V}(u|v)\ent(Z|U=u)\right)=\sum_{v\in{\cV}} \Prob_{{V}}(v)\left(\ent(Z|V=v)-\sum_{u\in\cU}\Prob_{U|V}(u|v)\ent(Z|U=u)\right),\mbox{ and}\label{eq:card4}\\
&|\tilde{\cV}|\leq |\cX|+2.\label{eq:card5}
\end{align}
In the above, constraints.~\eqref{eq:card1} and~\eqref{eq:card2} ensure that $\MI(Y;\tilde{V})$ equals $\MI(Y;V)$,~\eqref{eq:card3} and~\eqref{eq:card4} ensure that $\MI(\tilde{U};Y|\tilde{V})-\MI(\tilde{U};Z|\tilde{V})$ equals $\MI(U;Y|V)-\MI(U;Z|V)$, and Eq.~\eqref{eq:card5} follows from Caratheodory's theorem (\emph{c.f.}~\cite[Lemma 3]{AhlswedeK:75}) as Eqs.~\eqref{eq:card1}-\eqref{eq:card4} imply at most $|\cX|+2$ constraints on $\Prob_{\tilde{V}}$. Note that the number of constraints in our setting is one less than that in~\cite{CsiszarK:78} as we do not require  $\MI(Z;\tilde{V})$ to equal  $\MI(Z;V)$. Next, using a similar reasoning, the next step is to show that there exists an auxiliary variable $\hat{U}$ that takes values in a set $\hat{\cU}$, is jointly distributed with $\tilde{V},X,Y$, and $Z$ such that $\Prob_{\tilde{V}\hat{U}XYZ}(v,u,x,y,z)=\Prob_{\tilde{V}}(v)\Prob_{\hat{U}|\tilde{V}}(u|v)\Prob_{X|U}(x|u)\Prob_{YZ|X}(y,z|x)$ for every $(v,u,x,y,z)\in\tilde{\cV}\times\hat{\cU}\times\cX\times\cY\times\cZ$, where, for each $v\in\tilde{\cV}$, $\Prob_{\hat{U}|\tilde{V}}$ satisfies the following constraints:

\begin{align}
&\sum_{u\in\hat{\cU}}\Prob_{\hat{U}|\tilde{V}}(u|v)\Prob_{X|U}(x|u)=	\sum_{u\in\cU}\Prob_{U|\tilde{V}}(u|v)\Prob_{X|U}(x|u)=\Prob_{X|\tilde{V}}(x|v)\label{eq:card6},\\
&\sum_{u\in\cU}\Prob_{\hat{U}|\tilde{V}}(u|v)\ent(Y|U=u)=\sum_{u\in\cU}\Prob_{U|\tilde{V}}(u|v)\ent(Y|U=u),\label{eq:card7}\\
&\sum_{u\in\cU}\Prob_{\hat{U}|\tilde{V}}(u|v)\ent(Z|U=u)=\sum_{u\in\cU}\Prob_{U|\tilde{V}}(u|v)\ent(Z|U=u),\label{eq:card8}\mbox{ and}\\
&\left|\set{u\in\hat{\cU}:\Prob_{\hat{U}|\tilde{V}}(u|v)>0}\right|\leq |\cX|+1.\label{eq:card9}
\end{align}
Here, constraint~\eqref{eq:card6} ensures consistency of the marginals of $\Prob_{\tilde{V}\hat{U}XYZ}$ and $\Prob_{\tilde{V}\tilde{U}XYZ}$ with respect to $(\tilde{V},X,Y,Z)$,  constraints~\eqref{eq:card7}  and~\eqref{eq:card8} (along with~\eqref{eq:card6}) ensure that $\MI(\hat{U};Y|\tilde{V})-\MI(\hat{U};Z|\tilde{V})$ equals $\MI(\tilde{U};Y|\tilde{V})-\MI(\tilde{U};Z|\tilde{V})$, and Eq.~\eqref{eq:card9} again follows from Caratheodory's theorem as Eqs.~\eqref{eq:card6}-\eqref{eq:card8} imply at most $|\cX|+1$ constraints on $\Prob_{\hat{U}|\tilde{V}}(\cdot|v)$. Finally, summing the bound from~\eqref{eq:card9} over all $v\in\tilde{\cV}$, we obtain that it suffices to let $|\hat{\cU}|$ be at most $\left(|\cX|+1\right)\left(|\cX|+2\right)$. This completes the proof of the converse.
\end{proof}

\subsection{Discussions}
\subsubsection{Plausible deniability vs Secrecy}
In the following discussion, we compare the capacity region $\cR_{\bm}$ to rate regions for two standard information-theoretic secrecy problems -- the Wire-Tap Channel~\cite{Wyner:75} and Broadcast Channel with Confidential messages~\cite{CsiszarK:78} (see Figure~\ref{fig:WTCBCC}). To this end, we first adapt the following definitions from~\cite{Wyner:75,CsiszarK:78}.

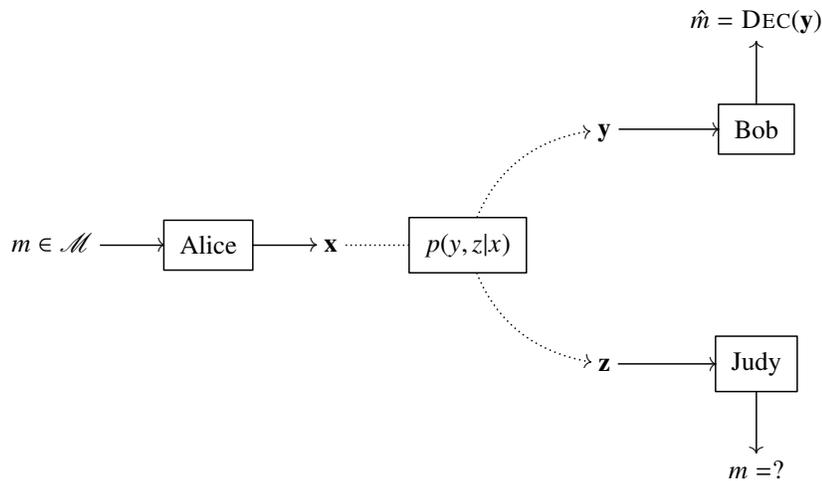
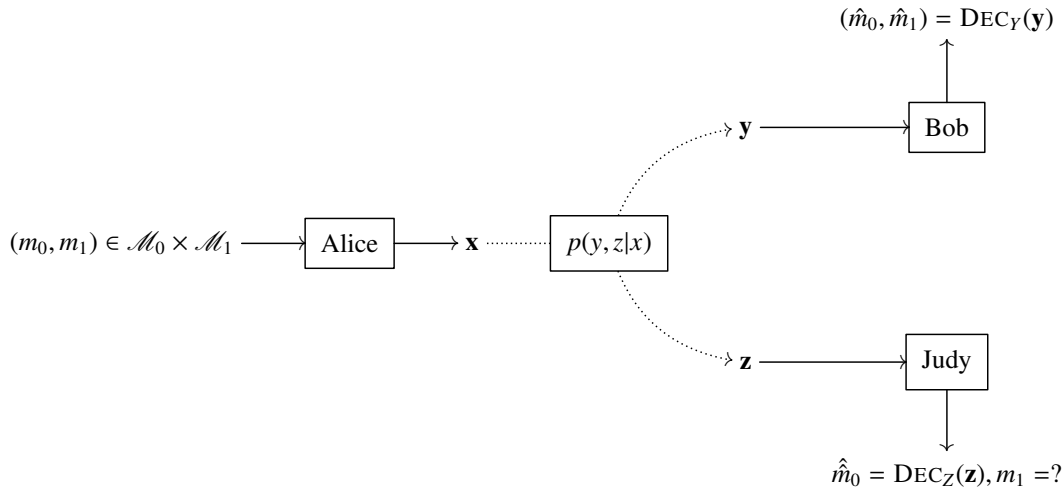
\begin{figure*}[h]
\centering
\begin{subfigure}{\textwidth}
\begin{displaymath}
\scalebox{1}{
\xymatrix{
&&&&& \hat{\bm}=\Dec(\by)\\
& &&&\by \ar[r]&
 *++[F]{\mbox{Bob}}\ar[u] &
 \\
\bm\in\cM\ar[r] &
 *++[F]{\mbox{Alice}}\ar[r]& \bx  \ar@{.}[r]&*++[F]{p(y,z|x)}\ar@/^1pc/@{.>}[ru]\ar@/_1pc/@{.>}[rd] && \\
 &&&&\bz\ar[r]&*++[F]{\mbox{Judy}}\ar@{->}[d]\\
&&&&&\bm=?}
}
\end{displaymath}\caption{The Wire-Tap Channel}\label{fig:WTC}
\end{subfigure}\\
\begin{subfigure}{\textwidth}
\begin{displaymath}
 \scalebox{1}{
\xymatrix{
&&&&& (\hat{\bm}_0,\hat{\bm}_1)=\Dec_Y(\by)\\
& &&&\by \ar[r]&
 *++[F]{\mbox{Bob}}\ar[u] &
 \\
(\bm_0,\bm_1)\in\cM_0\times\cM_1\ar[r] &
 *++[F]{\mbox{Alice}}\ar[r]& \bx  \ar@{.}[r]&*++[F]{p(y,z|x)}\ar@/^1pc/@{.>}[ru]\ar@/_1pc/@{.>}[rd] && \\
 &&&&\bz\ar[r]&*++[F]{\mbox{Judy}}\ar@{->}[d]\\
&&&&&\hat{\hat{\bm}}_0=\Dec_Z(\bz), \bm_1=?}
}
\end{displaymath}\caption{Broadcast Channel with Confidential Messages}\vspace{2em}\label{fig:BCC}
\end{subfigure}

\caption{In the Wire-Tap Channel problem (first introduced by~\cite{Wyner:75} and explored further in~\cite{CsiszarK:78}), the goal for Alice is to transmit a confidential message $\bm$ to the legitimate receiver Bob while ensuring that the ``leakage'' to the eavesdropper Judy (measured through the rate of equivocation) is smaller than a threshold. The capacity region for this problem (see Definition~\ref{def:REQ} exhibits a tradeoff between the message rate $R$ and the equivocation rate $R_e$. The Broadcast Channel with Confidential messages setup (introduced by~\cite{CsiszarK:78}) generalizes the Wire-Tap Channel model to include a ``public'' message $\bm_0$ that is meant to be decoded by both Bob and Judy. Similarly to the Wire-Tap Channel, this setup also includes a confidential message $\bm_1$ that is meant to be decoded by only Bob while ensuring that the leakage to Judy is smaller than a threshold. In general, the capacity region for this setup exhibits a tradeoff between three parameters -- the rate of the public message $R_0$, the rate of the confidential message $R_1$, and the equivocation rate. In our discussion, we only consider a two-dimensional projection of this region (see Definition~\ref{def:RBCC}) to the set of $(R_0,R_1)$ pairs that ensure that the equivocation about the message $\bm_1$ is arbitrarily close to the entropy of $\bm_1$. The reader is referred to~\cite{BlochB:11} for an excellent introduction to these and other information-theoretic security problems.}\label{fig:WTCBCC}
\end{figure*}

\begin{definition}[Rate-Equivocation Region]\label{def:REQ} For a channel $\Prob_{Y,Z|X}$, the rate-equivocation region $\REQ$ is the set of all non-negative $(R,R_e)$ pairs such for any $\epsilon>0$ and large enough block-length $n$, there exists a code for the Wire-Tap Channel problem (Figure~\ref{fig:WTC}) when the message $|\cM|\geq 2^{nR}$, $\Probc_{\bM(\bm)}=1/|\cM|$ for each $\bm\in\cM$, $\Probc_{\bM,\bX,\bY}(\bm\neq\hat{\bm})<\epsilon$, and $\ent(\bM|\bZ)\geq nR_e$.
\end{definition}

\begin{definition}[Sum Capacity with Confidential and Public messages]\label{def:RBCC} For a channel $\Prob_{Y,Z|X}$, the sum capacity region with confidential and public messages $\RBCC$ is the set of all non-negative $(R,R_1)$ pairs with $R\geq R_1$ for which, given any $\epsilon>0$, for a large enough blocklength $n$, there exists a code for the Broadcast Channel with Confidential Messages setup (Figure~\ref{fig:BCC}) with $|\cM_0|\geq 2^{n(R-R_1)}$, $|\cM_1|\geq 2^{nR_1}$, $\Probc_{\bM_0,\bM_1}(\bm_0,\bm_1)=1/{|\cM_0| |\cM_1|}$ for each $(\bm_0,\bm_1)\in\cM_0\times\cM_1$,  $\Probc_{\bM_0,\bM_1,\bX,\bY,\bZ}\left((\hat{\bM}_0,\hat{\hat{\bM}}_0,\hat{\bM}_1)\neq(\bM_0,\bM_0,\bM_1)\right)<\epsilon$ and $\ent(\bM_1|\bZ)\geq n R_1-\epsilon$. 
\end{definition}

 We note that in the Message Deniability setting, the existence of $(\delta,D)$-plausibly deniable faking procedure implies that the equivocation of $\bM$ given $\bZ$ is no smaller than $D-O(\sqrt{\delta})$.
 \begin{proposition} \label{prop:REQ} Let $\Fak{\bM}$ be $(\delta,D)$-plausibly deniable for $\bM$ given observation $\bZ$ and satisfy $\Fak{\bM}-\bM-\bZ$. Then, there exists $\mu$ depending only on $\Prob_{Z|X}$ such that $$\ent(\bM|\bZ)\geq nD-n\mu\sqrt{\delta}-2\delta.$$
\end{proposition}
\begin{proof}The above proposition is a direct consequence of Lemma~\ref{lem:convmessage}. Specifically, note that there exists $\lambda=\lambda(\Prob_{Z|X})$ such that
\begin{align*}
\ent(\bM|\bZ)&\geq\ent(\bM|\bZ)+\MI(\bM;\bZ|\Fak{\bM})-\delta-n\lambda\sqrt{\delta}\\
&= \ent(\bM|\bZ)+\ent(\bM|\Fak{\bM})-\ent(\bM|\bZ,\Fak{\bM})-\delta-n\lambda\sqrt{\delta}\\
&\geq  \ent(\bM|\Fak{\bM}) -\delta-n\lambda\sqrt{\delta}\\
& = \ent(\Fak{\bM}|\bM)+ \ent(\bM)-\ent(\Fak{\bM})-\delta-n\lambda\sqrt{\delta}\\
&\geq nD -2\delta -2n\lambda\sqrt{\delta}.
\end{align*}	\end{proof}

The above proposition leads to the following corollary.
\begin{corollary}\label{cor:inclusion}$\RBCC\subseteq\cR_{\bm}\subseteq\REQ$. 
\end{corollary}
\begin{proof}
	As proved in~\cite{CsiszarK:78}, $\RBCC$ is the set of all $(R,R_1)$ pairs such that there exist random variables $V$ and $U$ satisfying $V-U-X-(Y,Z)$ and 
	\begin{align*}
	0&\leq R\leq \min\{\MI(V;Y),\MI(V;Z)\}+\MI(U;Y|V)-\MI(U;Z|V)\\
	0&\leq R_0\leq \min\{\MI(V;Y),\MI(V;Z)\}	
	\end{align*}
 The first inclusion, $\RBCC\subseteq\cR_{\bm}$, follows directly by comparing our characterization of $\cR_{\bm}$ with the above capacity expression. Note that in the setting of~\cite{CsiszarK:78}, the public message of rate $R_0$ is intended to be decoded by both the receivers, while in our achievability proof of Theorem~\ref{thm:message}, we require that it be decoded only by Bob. This allows us to operate with public message rates as high as $\MI(V;Y)$, rather than $\min\set{\MI(V;Y),\MI(V;Z)}$ as in~\cite{CsiszarK:78}.
	Next, applying Proposition~\ref{prop:REQ} to a sequence of codes with $\delta$ approaching zero, we obtain that every $(R,D)\in\cR_{\bm}$ also lies in $\REQ$. \end{proof}

The following example illustrates that both inclusions in the above corollary may be strict.
\begin{example}[Binary Erasure Eavesdropper]\label{ex:BEC} Consider the example of Figure~\ref{fig:examplebec}. Let $\cX=\cY=\set{0,1}$, $\cZ=\set{0,\perp,1}$, and $$\Prob_{YZ|X}(yz|x)=\left\{\begin{array}{ll}1-p& \mbox{if } (y,z)=(x,x),\\ p &\mbox{if } (y,z)=(x,\perp),\mbox{ and}\\0&\mbox{otherwise.}\end{array}\right.$$ 
% It is clear that the zero-information variables for both $\Prob_{Z|X}$ not $\Prob_{Z|Y}$ are trivial, {\em i.e.}, they equal $X$ and $Z$ respectively. Thus, the maximum rate of deniability in both the transmitter and receiver deniability equals zero. In other words, $\cR_\bx=\cR_\by=[0,1]\times\{0\}$. On the other hand, 
As this is a degraded channel, it suffices to let the variable $U$ in Theorem~\ref{thm:message} be equal to $X$. Further, using the fact that $Y=X$ and $H(X|Z,V)=pH(X|V)$ (as the channel from $X$ to $Z$ is a Binary Erasure Channel with erasure probability $p$), we obtain the following characterisation for $\cR_\bm$. $\cR_\bm$ is the set of $(R,D)$ pairs such that there exists a random variable $V$ with $V-X-Z$, 
\begin{align*}
&0\leq R\leq \ent(X)-(1-p)\ent(X|V),\mbox{ and}\\
&0\leq D\leq \min\{R, p\ent(X|V)\}.
\end{align*} 
Let $\alpha_{X,V}=H(X|V)/H(X)$. Thus, $\ent(X)-(1-p)\ent(X|V)=(1-\alpha_{X,V}(1-p))\ent(X)$, and $p\ent(X)=p\alpha_{X,V} \ent(X)$. Note that $\alpha_{X,V}$  may take any value in the interval $[0,1]$ and the maximum value of $\ent(X)$ equals $1$. Thus, $\cR_\bm$ consists of $(R,D)$ pairs such that for some $\alpha\in[0,1]$, $0\leq R\leq (1-\alpha(1-p))$ and $0\leq D\leq \min\{R,\alpha p\}$. Simplifying further, we conclude that the region $\cR_\bm$ consists of $(R,D)$ pairs such that  
\begin{align*}
&0\leq R\leq 1\\
&0\leq D\leq \min\set{\frac{p(1-R)}{1-p},R}.
\end{align*}
	We next compare this region with the regions $\RBCC$ and $\REQ$. 	For the channel considered in this example, the Rate-Equivocation region consists of all $(R,R_e)$ pairs satisfying
	\begin{align*}
&0\leq R\leq 1\\
&0\leq R_e\leq \min\{p,R\}.\\
\end{align*}
Next, the region $\RBCC$ consists of all $(R,R_1)$ pairs satisfying 
	\begin{align*}
&0\leq R\leq 1\\
&0\leq R_1\leq \min\set{\frac{p(1-p-R)}{1-2p},R}.\\
\end{align*}
Comparing the above regions, it is evident that the inclusion relation in Corollary~\ref{cor:inclusion} may be strict. The  plot shown in Figure~\ref{fig:examplecomparision} compares these  regions.
\begin{figure}[h]
\begin{center}
\scalebox{1}{\begin{tikzpicture}[baseline, scale=0.75]
\fill[draw opacity=0.4, pattern color = myorange!50, pattern=north east lines] (0,0) -- (4,4) -- (10,4) -- (10,0) -- cycle;
\fill[fill opacity=0.75, white] (0,0) -- (4,4) -- (10,0) -- cycle;
\fill[draw opacity=0.9, pattern color = mygreen!80, pattern=north west lines] (0,0) -- (4,4) -- (10,0) -- cycle;
\fill[fill opacity=.7,myblue!80] (0,0) -- (4,4) -- (6,0) -- cycle;

  \begin{scope}
\draw[->, thick] (0,0) -- (11,0) node[below right] {{$R$}};
\draw[->, thick] (0,0) -- (0,5) node[above left] {{ $D$, $R_1$, $R_e$}};

\draw[dashed] (4,0) -- (4,4) node[above] {};
\draw[dashed] (0,4) -- (4,4) node[above] {};
\draw[red,thick,-] (10.01,.01) -- (9.99,-.01) ;
\draw[red,thick,-] (9.99,.01) -- (10.01,-.01)  node[black,below] {{ $\ \ 1$}};
\draw[red,thick,-] (4.01,.01) -- (3.99,-.01) ;
\draw[red,thick,-] (3.99,.01) -- (4.01,-.01)  node[black,below] {{ $\ \ p$}};
\draw[red,thick,-] (.01,4.01) -- (-.01,3.99) ;
\draw[red,thick,-] (.01,3.99) -- (-.01,4.01)  node[black,left] {{ $\ \ p$}};
\draw[red,thick,-] (6.01,.01) -- (5.99,-.01) ;
\draw[red,thick,-] (5.99,.01) -- (6.01,-.01)  node[black,below] {{ $\ \ 1-p$}};

\draw[-, line width = 0.4mm, myorange] (0,0) -- (4,4);
\draw[-, line width = 0.4mm, myorange] (4,4) -- (10,4);
\draw[-, line width = 0.4mm, myorange] (10,4) -- (10,0);
\draw[-, line width = 0.4mm, mygreen] (0,0) -- (4,4);
\draw[-, line width = 0.4mm, mygreen] (4,4) -- (10,0);

\draw[-, line width = 0.2mm, myblue] (0,0) -- (4,4);
\draw[-, line width = 0.2mm, myblue] (4,4) -- (6,0);

%\draw[red,-] (6,3) -- (7,3)  node[black,right] {{ $\REQ$}};
\draw[line width =0.2mm, <-,myorange] (9.5,3.5) to[out=45,in=180]   (11,3) node[right] {$\REQ$};
\draw[line width =0.2mm, <-,myblue] (3,2.5) to[out=135,in=0]   (2.5,4.5) node[left] {$\RBCC$};
\draw[line width =0.2mm, <-,mygreen] (5,2.5) to[out=135,in=0]   (5,5) node[left] {$\cR_{\bm}$};

\end{scope}
\end{tikzpicture}}
\end{center}
\caption{Comparision of $\cR_{\bm}$ with $\RBCC$ and $\REQ$ in Example~\ref{ex:BEC}.}\label{fig:examplecomparision}
\end{figure} 
\end{example}
\subsubsection{Rate of deniability as the Equivocation rate} Even though we define the rate of deniability as an operational property of the faking procedure, surprisingly, it also has a rough interpretation as the rate of equivocation given  the eavesdropper channel output as well as the fake message. This is especially interesting in light of Example~\ref{ex:BEC} that shows that the rate of deniability may be strictly smaller than the equivocation rate at the eavesdropper in the Wire-Tap Channel setting. The following proposition states this property formally.
\begin{proposition}\label{prop:msgDEQ} Let $\Fak{\bM}$ be $(\delta,D)$-plausibly deniable for $\bM$ given observation $\bZ$ and satisfy $\Fak{\bM}-\bM-\bZ$. Then, there exists $\mu\geq 0$ depending only on $\Prob_{YZ|X}$ such that $$nD-\delta-n\mu\sqrt{\delta}\leq\ent(\bM|\Fak{\bM},\bZ)\leq nD+\delta+n\mu\sqrt{\delta}.$$
\end{proposition}
\begin{proof}
Note that \begin{align*}
	\ent(\bM|\Fak{\bM},\bZ)&=\ent(\bM|\Fak{\bM})-\MI(\bM;\bZ|\Fak{\bM})\\
	&= \ent(\Fak{\bM}|\bM)-\ent(\Fak{\bM})+\ent(\bM)-\MI(\bM;\bZ|\Fak{\bM})\\
	& = nD-\ent(\Fak{\bM})+\ent(\bM)-\MI(\bM;\bZ|\Fak{\bM}).
\end{align*}
Applying Lemma~\ref{lem:convmessage} and the non-negativity of mutual information to the terms on the left hand side above gives the claimed result.\end{proof}

\section{Transmitter and Receiver deniability}\label{sec:codeword}
%ADD INTUITION -- USE EXAMPLE TO MOTIVATE Z.I.V.

Before formally proving Theorems~\ref{thm:transmitter} and~\ref{thm:rxachievability}, we introduce the notion of \emph{zero information variables} that is central to our discussion of the achievability proofs presented in this section. 
\subsection{Zero Information Variables}\label{sec:zeroinformation}For a random variable $W\sim \Prob_W$ and a channel $\Prob_{Z|W}$, we define the following relation: for $w_1,w_2\in{\cW}$, we say that $w_1\sim w_2$ if $\Prob_{Z|W}(z|w_1)=\Prob_{Z|W}(z|w_2)$, for all $z\in {\cZ}$. It is evident that this is an equivalence relation. Let $\cU_0$ represent the set of equivalence classes of this relation. We define the {\em zero-information} random variable $U_0$ of $W$ w.r.t. $\Prob_{Z|W}$ as a random variable taking values in ${\cU}_0$ and jointly distributed with $W$ and $Z$ such that $W\in U_0$ with probability $1$. For each $w\in\cW$, we will call the corresponding $u_0$ its {\em zero-information symbol}. 

Note that, $U_0$ is a function of $W$.  Intuitively, the zero information symbol $u_0$ of $w$ is the largest subset of $\cW$ such that each $w'\in u_o$ is statistically indistinguishable from $w$ given any $z\in\cZ$ with $\Prob_{Z|W}(z|w)>0$.  Figure~\ref{fig:exampleziv} shows an example of a zero-information variable.  Note that $U_0-W-Z$ (since $U_0$ is a function of $W$), $W-U_0-Z$ (by definition), and $\Prob_{Z|W}(z|w)=\Prob_{Z|W,U_0}(z|w,u_o)=\Prob_{Z|U_0}(z|u_o)$ if $u_o$ is the zero-information symbol of $w$. 

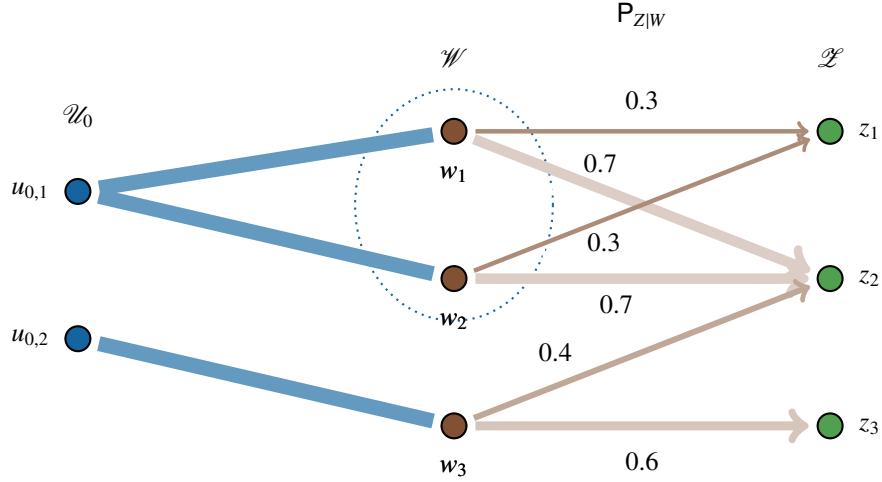
\begin{figure}[h]
\renewcommand{\labelstyle}{\textstyle}
\begin{center}
\renewcommand{\labelstyle}{\textstyle}
 \scalebox{1}{
 \begin{tikzpicture}[thick,
  every node/.style={draw,circle},
  fsnode/.style={fill=myblue},
  ssnode/.style={fill=mygreen},
  tsnode/.style={fill=myorange},
  every fit/.style={ellipse,draw,inner sep=-2pt,text width=2cm},
  ->,shorten >= 3pt,shorten <= 3pt
]

% the vertices of U_0
\begin{scope}[yshift = -8mm, start chain=going below,node distance = 16mm]
\foreach \i in {1,2}
  \node[fsnode,on chain] (u\i) [label=left: $u_{0,\i}$] {};
\end{scope}

% the vertices of W
\begin{scope}[xshift = 5cm, start chain=going below,node distance = 16mm]
\foreach \i in {1,2,3}
  \node[fsnode, white,on chain] (f\i) [label=below: $w_\i$] {};
\end{scope}

% the vertices of Z
\begin{scope}[xshift=10 cm,yshift=0cm,start chain=going below,node distance=16mm]
\foreach \i in {1,2,3}
  \node[ssnode,on chain] (s\i) [label=right: $z_\i$] {};
\end{scope}

% the set U
\node [white,fit=(u1) (u2), label=above:$\cU_0$] {};
\node [myblue, dotted, fit=(f1) (f2),label=above:$\cW$] {};
\node [white,fit=(s1) (s2), label=above:$\cZ$] {};
\node [white,fit=(f1) (s3), label=above:$\Prob_{Z|W}$] {};

% the set V
%\node [mygreen,fit=(s1) (s4),label=above:$\cZ$] {};

% the edges
\draw[color=myblue!66, line width = 2mm,-] (u1) -- (f1);
\draw[color=myblue!66, line width = 2mm,-] (u1) -- (f2);
\draw[color=myblue!66, line width = 2mm,-] (u2) -- (f3);

% the vertices of W again
\begin{scope}[xshift = 5cm, start chain=going below,node distance = 16mm]
\foreach \i in {1,2,3}
  \node[tsnode,on chain] [label=below: $w_\i$] {};
\end{scope}

\draw[color=myorange!66, line width = 0.6mm] (f1) -- (s1) node [color = black, line width = 0mm, draw = none, midway, above] {$0.3$};
\draw[color=myorange!28, line width = 1.4mm] (f1) -- (s2) node [color = black, draw = none, midway, above  left =3mm] {$0.7$};
\draw[color=myorange!66, line width = 0.6mm]  (f2) -- (s1) node [color = black, line width = 0mm, draw = none, midway, below left = 3mm] {$0.3$};
\draw[color=myorange!28, line width = 1.4mm]  (f2) -- (s2) node [color = black, draw = none, midway, below left] {$0.7$};
\draw[color=myorange!50, line width = 0.8mm] (f3) -- (s2) node [color = black, draw = none, midway, left=7mm] {$0.4$};
\draw[color=myorange!33, line width = 1.2mm]  (f3) -- (s3) node [color = black, draw = none, midway, below] {$0.6$};
\end{tikzpicture}
}
\end{center}
\caption{Let $W$ and $Z$ be random variables distributed on $\cW=\set{w_1,w_2,w_3}$ and $\cZ=\set{z_1,z_2,z_3}$ respectively with $\Prob_{Z|W}$ specified according to the edge labels in the above figure. Notice that $w_1$ and $w_2$ are indistinguishable to an observer who has access to only $Z$ as $\Prob_{Z|W}(z|w_1)=\Prob_{Z|W}(z|w_2)$ for every $z\in\cZ$. Hence, the zero-information variable for the distribution $\Prob_{Z|W}$ takes the values $u_{0,1}\equiv\{w_1,w_2\}$ and $u_{0,2}\equiv\{w_3\}$.}\label{fig:exampleziv}
\end{figure}

In our achievability proofs for Transmitter and Receiver deniability, the use of zero-information variables considerably simplifies the proof. In particular, we argue that the rate regions claimed achievable in Theorems~\ref{thm:transmitter} and~\ref{thm:rxachievability} it suffices to consider zero information variables instead of the general class of auxiliary variables presented in the theorem statements.  The following lemma shows that such a choice does not lead to any loss of optimality.
\begin{lemma} \label{lem:zeroinformationvariable} Suppose $W-U-Z$ and $U-W-(V,Z)$ are Markov chains. Then $\MI(W;V|U) \leq \MI(W;V|U_0)$, 
where $U_0$ is the zero-information random variable of $W$ w.r.t. $\Prob_{Z|W}$.
\end{lemma}

\begin{proof}
We first show that the Markov chains $W-U-Z$ and $U-W-Z$ imply that $U_0$ must also be a function of $U$. To show this, it is enough to show that for $w_1,w_2\in \cW$ with $\Prob_W(w_1),\Prob_W(w_2)>0$, if there is a $z\in{\cZ}$ such that $\Prob_{Z|W}(z|w_1) \neq \Prob_{Z|W}(z|w_2)$, then for every $u\in{\cU}$ at least one of $\Prob_{U|W}(u|w_1)$ and $\Prob_{U|W}(u|w_2)$ must be zero. Suppose, to the contrary both
$\Prob_{U|W}(u|w_1),\Prob_{U|W}(u|w_2)>0$. Then
\begin{align*}
\Prob_{Z|W}(z|w_1) 
        &\overset{(a)}{=} \Prob_{Z|W,U}(z|w_1,u)\\
        &\overset{(b)}{=} \Prob_{Z|U}(z|u)\\
        &\overset{(c)}{=} \Prob_{Z|W,U}(z|w_2,u)\\
        &\overset{(d)}{=} \Prob_{Z|W}(z|w_2),
\end{align*}
where $(a)$ follows from the Markov chain $U-V-Z$ and the fact that $\Prob_{W}(w_1)\Prob_{U|W}(u|w_1)>0$; $(b)$ follows from the Markov chain $W-U-Z$; $(c)$ follows from the Markov chain $W-U-Z$ and the fact that $\Prob_{W}(w_2)\Prob_{U|W}(u|w_2)>0$; and $(d)$ follows from the Markov chain $U-W-Z$. But, this is a contradiction. 

Thus, $U_0$ is a function of $U$. From its definition, $U_0$ is a function of $W$. Hence,
\begin{align*}
\MI(W;V|U) &= \MI(W;V|U,U_0)\\
        &\leq \MI(U,W;V|U_0)\\
        &= \MI(W;V|U_0) + \MI(U;V|W,U_0)\\
        &\overset{(a)}{=} \MI(W;V|U_0) + \MI(U;V|W)\\
        &\overset{(b)}{=} \MI(W;V|U_0),
\end{align*}
where $(a)$ uses the fact that $U_0$ is a function of $W$ and $(b)$ follows from $U-W-V$ being a Markov chain.\end{proof}

%The above lemma shows that for achieving the rates claimed in Theorems~\ref{thm:transmitter} and~\ref{thm:rxachievability}, it suffices to restrict our auxiliary random variables to the corresponding zero information variables. This is critical to our proofs.
\subsection{Transmitter Deniability}
We begin our proof for Theorem~\ref{thm:transmitter} by stating two lemmas that lead to our converse arguments.  The following lemma mirrors Lemma~\ref{lem:convmessage} from the message deniability setting and derives necessary conditions for any faking procedure to be plausible with respect to the eavesdropper's observation.  In particular, we show that for any plausibly deniable faking procedure, the true codeword and the eavesdropper observation must be nearly conditionally independent given the fake codeword. Further,  the joint distribution of the  true and fake codewords must be such that it allows exchanging $\bM$ for $\Fak{\bM}$ (and \emph{vice versa}) does not changes entropic terms involving these by at most $\delta$.
\begin{lemma}\label{lem:xzgivenf}
Let $\Fak{\bX}$ be $(\delta,D)$-plausibly deniable for $\bX$ given observation $\bZ$ and satisfy $\Fak{\bX}-\bX-\bZ$. Then, there exists a constant $\kappa$ depending only on $\Prob_{Z|X}$ such that 
\begin{align*}
&\MI(\bX;\bZ|\Fak{\bX})\leq n\kappa\sqrt{\delta},\\
&\left|\ent(\bX|\Fak{\bX})-\ent(\Fak{\bX}|\bX)\right|\leq n\kappa\sqrt{\delta},\\
&\left|\ent(\bX)-\ent(\Fak{\bX})\right|\leq n\kappa\sqrt{\delta}, \mbox{ and}\\
&\left|\ent(\bX|\Fak{\bX},\bM)-\ent(\Fak{\bX}|\bX,\Msg(\Fak{\bX}))\right|\leq n\kappa\sqrt{\delta}.	
\end{align*}
\end{lemma}

\begin{proof}
We explicitly only prove the first inequality. The other inequalities follow from a similar reasoning. 
\begin{align*}
\lefteqn{\MI(\bX;\bZ|\Fak{\bX})=\ent(\bZ|\Fak{\bX})-\ent(\bZ|\bX)}\\
%&={\mathlarger\sum_{\bz\in\cZ^n,\bx\in\cX^n}}\left[\Probc_{\bZ,\Fak{\bX}}(\bz,\bx)\log{\frac{\Probc_{\Fak{\bX}}(\bx)}{\Probc_{\bZ,\Fak{\bX}}(\bz,\bx)}} -\Probc_{\bZ,\bX}(\bz,\bx)\log{\frac{\Probc_{\bX}(\bx)}{\Probc_{\bZ,\bX}(\bz,\bx)}}\right]\\
&={\mathlarger\sum_{(\bz,\bx):\Probc_{\bZ,\Fak{\bX}}(\bz,\bx)>0}}\Probc_{\bZ,\Fak{\bX}}(\bz,\bx)\log{\frac{\Probc_{\Fak{\bX}}(\bx)}{\Probc_{\bZ,\Fak{\bX}}(\bz,\bx)}} \eqbr - {\mathlarger\sum_{(\bz,\bx):\Probc_{\bZ,\bX}(\bz,\bx)>0}}\Probc_{\bZ,\bX}(\bz,\bx)\log{\frac{\Probc_{\bX}(\bx)}{\Probc_{\bZ,\bX}(\bz,\bx)}}\\
&= {D}\left(\Probc_{\Fak{\bX}}||\Probc_{\bX}\right)- {D}\left(\Probc_{\bZ,\Fak{\bX}}||\Probc_{\bZ,\bX}\right)\eqbr+{\mathlarger\sum_{(\bz,\bx):\Probc_{\Fak{\bZ},\bX}(\bz,\bx)>0}}\Probc_{\bZ,\Fak{\bX}}(\bz,\bx)\log{\frac{\Probc_{\bX}(\bx)}{\Probc_{\bZ,\bX}(\bz,\bx)}}\eqbr - {\mathlarger\sum_{(\bz,\bx):\Probc_{\bZ,\bX}(\bz,\bx)>0}}\Probc_{\bZ,\bX}(\bz,\bx)\log{\frac{\Probc_{\bX}(\bx)}{\Probc_{\bZ,\bX}(\bz,\bx)}}\\
&\overset{(a)}{=} {D}\left(\Probc_{\Fak{\bX}}||\Probc_{\bX}\right)- {D}\left(\Probc_{\bZ,\Fak{\bX}}||\Probc_{\bZ,\bX}\right)\eqbr+{\mathlarger\sum_{(\bz,\bx):\Probc_{\bZ,\bX}(\bz,\bx)>0}}\Probc_{\bZ,\Fak{\bX}}(\bz,\bx)\log{\frac{\Probc_{\bX}(\bx)}{\Probc_{\bZ,\bX}(\bz,\bx)}}\eqbr -{\mathlarger\sum_{(\bz,\bx):\Probc_{\bZ,\bX}(\bz,\bx)>0}}\Probc_{\bZ,\bX}(\bz,\bx)\log{\frac{\Probc_{\bX}(\bx)}{\Probc_{\bZ,\bX}(\bz,\bx)}}\\
&= {D}\left(\Probc_{\Fak{\bX}}||\Probc_{\bX}\right)- {D}\left(\Probc_{\bZ,\Fak{\bX}}||\Probc_{\bZ,\bX}\right)\eqbr+{\mathlarger\sum_{(\bz,\bx):\Probc_{\bZ,\bX}(\bz,\bx)>0}}\left[\Probc_{\bZ,\Fak{\bX}}(\bz,\bx)-\Probc_{\bZ,\bX}(\bz,\bx)\right]\log{\frac{1}{\prod_{i=1}^n\Prob_{Z|X}(z_i|x_i)}}\\
&\overset{(b)}{\leq} \delta+n\sqrt{2\delta}\max_{(z,x):\Prob_{Z|X}(z|x)>0}\log{\frac{1}{\Prob_{Z|X}(z|x)}}.
\end{align*}
In the above, step $(a)$ uses the fact that $\set{(\bz,\bx):\Probc_{\Fak{\bZ},\bX}(\bz,\bx)>0}\subseteq\set{(\bz,\bx):\Probc_{\bZ,\bX}(\bz,\bx)>0}$ (as ${D}\left(\Probc_{\bZ,\Fak{\bX}}||\Probc_{\bZ,\bX}\right)<\delta<\infty$). In step $(b)$, we use the fact that $\Fak{\bX}$ is $(\delta,D)$-plausibly deniable $\bX$ for $\bZ$. The bound on the first term follow from definition, the second from non-negativity of K-L divergence, while the last term is bounded by applying Pinsker's inequality.\end{proof}

The following lemma follows from a standard chain of information inequalities with Lemma~\ref{lem:xzgivenf} as a starting point and single-letterizing the resulting expressions. 
\begin{lemma}\label{lem:degradedx}
 Let $\cC$ be an $(\epsilon,R)$-reliable code of blocklength $n$ for a channel $\Prob_{YZ|X}$, and let $\Fak{\bX}$ be $(\delta,D)$-plausibly deniable for $\bX$ given observation $\bZ$ and satisfy $\Fak{\bX}-\bX-\bZ$. Then, there exists random variables $U,X,Y$, and $Z$ satisfying $U - X - (Y, Z)$ and a constant $\converseeps=\converseeps(\epsilon,\delta)>0$ satisfying $\lim_{(\epsilon,\delta)\to (0,0)}\gamma=0$ such that
\begin{align*}
&R\leq \MI(X;Y)+\converseeps\mbox{, }\\
&D\leq \MI(X;Y|U)+\converseeps\mbox{, and}\\
&\MI(X;Z|U)\leq\converseeps.
\end{align*}
\end{lemma}

\begin{proof}
Note that $\bY-\bX-\Fak{\bX}$. We use Lemma~\ref{lem:xzgivenf} below.
\begin{align*}
nD&\leq \ent(\Msg(\Fak{\bX})|\bX)\\
&= \ent(\Fak{\bX}|\bX)-\ent(\Fak{\bX}|\bX,\Msg(\Fak{\bX}))\\
&\overset{(a)}{\leq} \ent(\bX|\Fak{\bX})-\ent(\bX|\Fak{\bX},\Msg(\bX))+2n\kappa\sqrt{\delta}\\
&\leq \ent(\bX|\Fak{\bX})-\ent(\bX|\bY,\Fak{\bX},\Msg(\bX))+2n\kappa\sqrt{\delta}\\
&= \ent(\bX|\Fak{\bX})-\ent(\bX|\bY,\Fak{\bX})+\MI(\bX;\Msg(\bX)|\bY,\Fak{\bX})+2n\kappa\sqrt{\delta}\\
%&= \MI(\bX;\bY|\Fak{\bX})+\ent(\bX|\bY,\Fak{\bX})\\
%&\qquad\mbox{} - \ent(\bX|\Fak{\bX},\Msg(\bX))+O(n\sqrt{\delta})\\
%&\leq \MI(\bX;\bY|\Fak{\bX})+\MI(\Msg(\bX);\bX|\bY,\Fak{\bX})+O(n\sqrt{\delta})\\
%&\leq \MI(\bX;\bY|\Fak{\bX})+\ent(\Msg(\bX)|\bY,\Fak{\bX})+O(n\sqrt{\delta})\\
&\overset{(b)}{\leq} \MI(\bX;\bY|\Fak{\bX})+n\epsilon+2n\kappa\sqrt{\delta}\\
&\overset{(c)}{\leq} \MI(\bX;\bY|\Fak{\bX})-\MI(\bX;\bZ|\Fak{\bX})+n\epsilon+3n\kappa\sqrt{\delta}\\
&= \sum_{i=1}^n \Big[ \MI(\bX;Y_i|\Fak{\bX},Y^{i-1})- \MI(\bX;Z_i|\Fak{\bX},Z_{i+1}^n)\Big]+n\epsilon+3n\kappa\sqrt{\delta}\\
&\overset{(d)}{=} \sum_{i=1}^n \Big[ \MI(\bX;Y_i|\Fak{\bX},Y^{i-1},Z_{i+1}^n)- \MI(\bX;Z_i|\Fak{\bX},Y^{i-1},Z_{i+1}^n)\Big] \eqbr +n\epsilon+3n\kappa\sqrt{\delta}\\
&= 	\sum_{i=1}^n\Big[ \ent(Y_i|\Fak{\bX},Y^{i-1},Z_{i+1}^n)-\ent(Y_i|\bX,\Fak{\bX},Y^{i-1},Z_{i+1}^n) \eqbr - \ent(Z_i|\Fak{\bX},Y^{i-1},Z_{i+1}^n)+\ent(Z_i|\bX,\Fak{\bX},Y^{i-1},Z_{i+1}^n)\Big] \\
&\quad  +n\epsilon+3n\kappa\sqrt{\delta}\\
&\overset{(e)}{=} \sum_{i=1}^n\Big[ \ent(Y_i|\Fak{\bX},Y^{i-1},Z_{i+1}^n)-\ent(Y_i|X_i,\Fak{\bX},Y^{i-1},Z_{i+1}^n) \eqbr - \ent(Z_i|\Fak{\bX},Y^{i-1},Z_{i+1}^n)+\ent(Z_i|X_i,\Fak{\bX},Y^{i-1},Z_{i+1}^n)\Big] \\
&\quad +n\epsilon+3n\kappa\sqrt{\delta}.
\end{align*}
In the above, $(a)$ and $(c)$ follow from Lemma~\ref{lem:xzgivenf}, $(b)$ is a consequence of Fano's inequality, $(d)$ is an application of Csisz\'ar's sum identity~\cite{ElGamalK:11}, and $(e)$ relies on the memoryless nature of the channel to argue that $(Y_i,Z_i)-X_i-(\Fak{\bX},Y^{i-1},Z_{i+1}^n,X^{i-1},X_{i+1}^n)$ is a Markov chain. Next, we let $U_i\triangleq(\Fak{\bX},Y^{i-1},Z_{i+1}^n)$,  and let $\share$ be a random variable independent of $(\bM,\Fak{\bX},\bX,\bY,\bZ)$ that is uniformly distributed over $[1 : n]$. Note that $U_i-X_i-(Y_i,Z_i)$ is a Markov chain. The above inequalities are continued further as
\begin{align*} 
nD&\leq\sum_{i=1}^n\Big[ \ent(Y_i|U_i)-\ent(Y_i|X_i,U_i)- \ent(Z_i|U_i)+\ent(Z_i|X_i,U_i)\Big] \eqbr +n\epsilon+3n\kappa\sqrt{\delta}\\
&= n\MI(X_\share;Y_\share|U_\share,\share)-n\MI(X_\share;Z_\share|U_\share,\share)+n\epsilon+3n\kappa\sqrt{\delta}.\numberthis\label{eq:txconverse1}\\
%&= \ent(\bY|\Fak{\bX})-\ent(\bY|\bX)+n\epsilon+O(n\sqrt{\delta})\\
%&= \sum_{i=1}^n\left[\ent(Y_i|\Fak{\bX},Y_{i+1}^n)-\ent(Y_i|X_i)\right]+n\epsilon+O(n\sqrt{\delta})\\
%&= \sum_{i=1}^n\left[\ent(Y_i|\Fak{\bX},Y_{i+1}^n,Z^{i-1})-\ent(Y_i|X_i)+\MI(Y_i;Z^{i-1}|\Fak{\bX},Y_{i+1}^n)\right]+n\epsilon+O(n\sqrt{\delta})\\
\intertext{Next, note that}
\lefteqn{\MI(X_\share;Z_\share|U_\share,\share)}\\
&=\frac{1}{n}\sum_{i=1}^n \MI(X_i;Z_i|\Fak{\bX},Y^{i-1},Z_{i+1}^{n})\\
 &=\frac{1}{n}\sum_{i=1}^n \left[\ent(Z_i|\Fak{\bX},Y^{i-1},Z_{i+1}^{n})-\ent(Z_i|X_i,	\Fak{\bX},Y^{i-1},Z_{i+1}^{n})\right]\\
&\overset{(a)}{=} \frac{1}{n}\sum_{i=1}^n \left[\ent(Z_i|\Fak{\bX},Y^{i-1},Z_{i+1}^{n})-\ent(Z_i|\bX,	\Fak{\bX},Z_{i+1}^{n})\right]\\
&\leq \frac{1}{n}\sum_{i=1}^n \MI(\bX;Z_i|\Fak{\bX},Z_{i+1}^{n})\\
&=\frac{1}{n}\MI(\bX;\bZ|\Fak{\bX})\\
&\overset{(b)}{<}\kappa\sqrt{\delta}.
\end{align*}
In the above, $(a)$ follows by noting that for each $i$, $Z_i-X_i-(X^{i-1},X_{i+1}^n,\Fak{\bX},Y^{i-1},Z_{i+1}^n)$ is a Markov chain due to the memoryless nature of the channel $\Prob_{Z|X}$ and $(b)$ follows from Lemma~\ref{lem:xzgivenf}. Defining random variables $(U,X,Y)$ with $\Probc_{U,X,Y}(u,x,y)=\Probc_{(U_\share,\share),X_\share,Y_\share}(u,x,y)$, we  obtain
\begin{align*}
D&\leq \MI(X;Y|U)+\epsilon+3\kappa\sqrt{\delta}.\\
\intertext{Notice that $\Probc_{Y|X}$ is the same as the channel transition probability $\Prob_{Y|X}$. Further, $U-X-(Y,Z)$ is a Markov chain and $\MI(X;Z|U)<\kappa\sqrt{\delta}$. Thus, $U$ satisfies the constraints from the lemma statement. Finally, we bound the rate  as follows.}
nR&=\ent(\bM)\\
&\overset{(a)}{\leq} \MI(\bX;\bY)+n\epsilon\\
&\leq \sum_{i=1}^n\MI(X_i;Y_i) + n\epsilon\\
&= n \MI(X_\share;Y_\share|\share)+n\epsilon\\
&\leq n \MI(\share,X_\share;Y_\share)+n\epsilon\\
&=n \MI(X_\share;Y_\share)+\MI(\share;Y_\share|X_\share)+n\epsilon\\
& \overset{(b)}{=} n \MI(X_\share;Y_\share)+n\epsilon\\
&= n \MI(X;Y)+n\epsilon.
\end{align*}
In the above, $(a)$ follows from Fano's inequality and $(b)$  from the fact that $\Probc_{Y_\share|(X_\share,\share)}(y|x,t)=\Prob_{Y|X}(y|x)$. Letting $\converseeps=\epsilon+3\kappa\sqrt{\delta}$ proves the lemma.\end{proof}

We are now ready to formally prove Theorem~\ref{thm:transmitter}. The converse essentially follows from the results that we have earlier in this section. Using these, we show that every achievable $(R,D)$ must satisfy  the upper bounds stated in the theorem for some choice of an auxiliary random variable $U$ satisfying $U-X-(Y,Z)$ and $Y-U-Z$. For the direct part of the proof, we prove the achievability of all $(R,D)$ that satisfy upper bounds provided by the theorem statement when $U$ is the zero information variable of $X$ with respect to $\Prob_{Z|X}$. We note that restricting  the choice of $U$ to be the zero information variable entails no loss in optimality (as shown in Lemma~\ref{lem:zeroinformationvariable}).

\begin{proof}[Proof of Theorem~\ref{thm:transmitter}] The converse for Theorem~\ref{thm:transmitter} follows by invoking Lemma~\ref{lem:degradedx} for a vanishing sequence of $\delta$'s and by applying standard continuity arguments from Lemma~\ref{lem:continuity} to show that any $(R,D)\in\cR_{\bx}$ must satisfy 
\begin{align}
&0\leq R\leq \MI(X;Y)\mbox{ and }\label{eq:txbound1} \\
&0\leq D\leq \MI(X;Y|U)\label{eq:txbound2}
\end{align} for some random variable $U$ satisfying the Markov chains $U-X-(Y,Z)$ and $X-U-Z$. Now, applying Lemma~\ref{lem:zeroinformationvariable},  we note that $ \MI(X;Y|U)\leq  \MI(X;Y|U_0)$, where $U_0\in\cU_0$ is the zero information variable of $X$ w.r.t. $\Prob_{Z|X}$. Further, by definition,  $|\cU_0|\leq |\cX|$. Thus, to describe the region given by Eqs.~\eqref{eq:txbound1} and~\eqref{eq:txbound2}, it suffices to consider auxiliary variables $U$ whose support is of size no larger than $|\cX|$.

We now give a proof sketch for the achievability of claimed rate region. Our achievability uses a superposition code for the broadcast channel $\Prob_{Y,Z|X}$. To this end,  choose random variables $(X,U)$ satisfying the conditions in the theorem with $U$ as the zero information variable of $X$ w.r.t. $\Prob_{Z|X}$. Recall that Lemma~\ref{lem:zeroinformationvariable} guarantees that there is no loss of optimality in choosing $U$ as the zero information variable of $X$ w.r.t. $\Prob_{Z|X}$. In the following, we prove the achievability of the rate pairs that lie on the boundary of the claimed region, \emph{i.e.}, we consider $(R,D)$ where $R=\MI(X;Y) - 2\epsilon$ and $D=\MI(X;Y|U) - \epsilon$. 

We consider a superposition code via a standard random coding argument.  For any $\epsilon>0$, first, we generate $\widetilde{\cC}=\{\bu(1),\ldots,\bu(2^{n(R-D)})\}$ by drawing $u_i(j)$ independently from the distribution $\Prob_U$ for each $i\in\{1,2,\ldots,n\}$ and $j\in\{1,2,\ldots,2^{n(R-D)}\}$. Next, for each $j\in\{1,2,\ldots,2^{n(R-D)}\}$, we generate a sub-code $\cC_{j}=\{\bx(j,1),\ldots,\bx(j,2^{nD})\}$ by drawing  $x_i(j,k)$ independently from the distribution $\Prob_{X|U}(\cdot|u_i(j))$ for each $i\in\{1,2,\ldots,n\}$ and $k\in\{1,2,\ldots,2^{nD}\}$. We then form the codebook $\cC=\{\bx^{(\cC)}(\bm):\bm\in\cM\}$ by taking the union  $\cup_{j\in\set{1,2,\ldots,2^{n(R-D)}}}\cC_{j}$. Finally, the faking procedure simply accepts the transmitted codeword (say, $\bx$) and outputs a uniformly drawn codeword from the sub-code that contains $\bx$ (say, $\cC_{j}$).

Since the reliability of the above code follows from standard arguments for superposition coding (see~\cite{ElGamalK:11} for example), we skip the detailed analysis here. The plausible deniability for the code follows directly from the construction by noting that for every $\bx\in\cC_{j}$, $\bu(j)$ is precisely the sequence of the zero information symbols of $\bx$ w.r.t. $\Prob_{Z|X}$. Thus, for any $\bx,\bx'\in\cC_{j}$ and $\bz\in\cZ^n$, $\Probc_{\bX|\bZ}(\bx|\bz)=\Probc_{\bX|\bZ}(\bx'|\bz)$.\end{proof}

\subsection{Receiver Deniability}
In this section, we give the proof of our achievability for Receiver Deniability in the physically degraded channel setting. As earlier, Lemma~\ref{lem:zeroinformationvariable} shows that the rate region claimed in Theorem~\ref{thm:rxachievability} is unchanged if $V$ is restricted to be the zero information variable of $Y$ with respect to $\Prob_{YZ|X}$. In the following, we prove the achievability of $(R,D)$ that satisfy the bounds in Theorem~\ref{thm:rxachievability} with respect to an auxiliary variable $V$  that is the zero information variable of $Y$ with respect to $\Prob_{YZ|X}$. %and comment on the tightness of this region in subsequent remarks.
\begin{proof}[Proof of Theorem~\ref{thm:rxachievability}] Let $\epsilon>0$, fix a blocklength $n$, set 
\begin{equation}R=\MI(X;Y)-\rho\label{eq:rxrate}\end{equation}
for some $\rho>\epsilon$, and $|\cM|=2^{nR}$. Let $V$  be the zero information variable of $Y$ with respect to $\Prob_{YZ|X}$. Our achievability uses a random coding argument. Consider the following codebook generation procedure and the corresponding faking procedure. 
\paragraph{Codebook generation} The codebook $\cC$ is a multiset $\{\bx^{(\cC)}(\bm):\bm\in\cM\}$ that is generated by drawing each $x^{(\cC)}_i(\bm)$ independently from the distribution $\Prob_X$. Let ${\Pr}_{\cC}$ be the probability distribution over the random generation of the codebook. 
\paragraph{Encoding} For a message $\bm\in\cM$, the encoder transmits $\bx^{(\cC)}(\bm)$.
\paragraph{Decoding} Upon receiving $\by$, the decoder looks for $\bm\in\cM$ such that $(\bx^{(\cC)}(\bm),\by)\in\AEN{X,Y}$. 
\paragraph{Faking procedure} Given $\by$, the faking procedure first generates the unique $\bv$ where, for each $i$, $v_i$ represents the zero information symbol of $y_i$ w.r.t. $\Prob_{Z|Y}$. Next, $\Fak{\bY}$ is drawn from $\cY^n$ according to the conditional distribution $\Probc_{\Fak{\bY}|\bV}=\Probc_{\bY|\bV}$. Note that the distribution $\Probc_{\bY|\bV}$ depends on both the codebook as well the channel.

\paragraph{Analysis} Note that $\Fak{\bY} - \bV - (\bY,\bX,\bZ)$ is a Markov chain. For similar reasons as in transmitter deniability, these ensure that the parameter $\delta$ is zero. To this end, we first observe that for any $(\by,\bv,\bz)\in\cY^n\times\cV^n\times\cZ^n$,
\begin{align*}
	\Probc_{\bY\bV\bZ}(\by,\bv,\bz)&\overset{(a)}{=} \Probc_{\bY\bV}(\by,\bv)\Probc_{\bZ|\bY}(\bz|\by)\\
	&\overset{(b)}{=}\Probc_{\bY|\bV}(\by|\bv)\Probc_{\bV}(\bv)\Prob_{\bZ|\bY}(\bz|\by)\\
	&\overset{(c)}=\Probc_{\bY|\bV}(\by|\bv)\Probc_{\bV}(\bv)\Prob_{\bZ|\bV}(\bz|\bv),
\intertext{and }
\Probc_{\Fak{\bY}\bV\bZ}(\by,\bv,\bz)&= {\mathlarger\sum_{\by'\in\cY^n}}\Probc_{\bY,\Fak{\bY}\bV\bZ}(\by',\by,\bv,\bz)\\
&\overset{(d)}{=} {\mathlarger\sum_{\by'\in\cY^n}}\Probc_{\bY|\bV}(\by'|\bv)\Probc_{\Fak{\bY}|\bV}(\by|\bv)\Probc_{\bV}(\bv)\Probc_{\bZ|\bY}(\bz|\by')\\
&\overset{(e)}{=} {\mathlarger\sum_{\by'\in\cY^n}}\Probc_{\bY|\bV}(\by'|\bv)\Probc_{\Fak{\bY}|\bV}(\by|\bv)\Probc_{\bV}(\bv)\Prob_{\bZ|\bV}(\bz|\bv)\\
&=\Probc_{\Fak{\bY}|\bV}(\by|\bv)\Probc_{\bV}(\bv)\Prob_{\bZ|\bV}(\bz|\bv)\\
&\overset{(f)}{=}\Probc_{\bY|\bV}(\by|\bv)\Probc_{\bV}(\bv)\Prob_{\bZ|\bV}(\bz|\bv)\\
&= \Probc_{\bY\bV\bZ}(\by,\bv,\bz).\\
\end{align*}
In the above, $(a)$ and $(d)$ follow from the dependence structure of the random variables $\bY,\Fak{\bY},\bV$, and $\bZ$, $(b)$ is a consequence of the channel being physically degraded, $(c)$ and $(e)$ are true since $V$ is the zero information variable of $Y$ w.r.t. $\Prob_{Y|Z}$, and $(f)$ is implied by the faking procedure used to generate $\Fak{\bY}$.
Thus, 
\begin{align*}
\lefteqn{\delta = \KL(\Probc_{\Fak{\bY}\bZ}||\Probc_{\bY\bZ})}\\
&\leq \KL(\Probc_{\Fak{\bY}\bV\bZ}||\Probc_{\bY\bV\bZ})	\\
&=0.
\end{align*}
Next, we analyze the rates $(R,D)$ that our code and faking procedure can achieve. Let $\alpha\in(0,1)$. The reliability analysis is similar to Shannon's channel coding theorem. Let $\cG_1\triangleq\{\cC:\Probc_{\bM,\bX,\bY}(\bM\neq\bMh)<\epsilon\}$ denote the class of codebooks that have an average error probability smaller than $\epsilon$. Following the standard proof of reliability of random codes, there exists $n_1=n_1(\alpha)$ such that as long as $R<\MI(X;Y)$ and $n>n_1$, 
\begin{equation}{\Pr}_{\cC}\left(\cG_1\right)\geq1- \alpha/4\label{eq:rxreliability}.\end{equation} 
In the following we assume that $\cC\in\cG_1$ and prove that, with a high probability over the codebook generation, the rate of deniability for our faking procedure is large enough for our theorem. To this end, the following chain of inequalities give a lower bound on $D$ for the code $\cC$.
\begin{align*}
\lefteqn{nD=\ent(\Dec(\Fak{\bY})|\bY)}\\
&=\ent(\Dec(\Fak{\bY})|\bV\,\bY)\numberthis\label{eq:ach-1}\\
&=\ent(\Dec(\Fak{\bY})|\bV)\numberthis\label{eq:ach0}\\
&= \ent(\Dec(\bY)|\bV)\numberthis \label{eq:ach1}\\
&\geq \MI(\Dec(\bY);\bX|\bV)\\& = \MI(\Dec(\bY),\bY;\bX|\bV)-\MI(\bX;\bY|\bV,\Dec(\bY))\\
& = \MI(\bX;\bY|\bV)-\MI(\bX;\bY|\bV,\Dec(\bY)) \\& \geq \MI(\bX;\bY|\bV)-n\epsilon.\numberthis\label{eq:ach2}
\end{align*}
In the above, Eq.~\eqref{eq:ach-1} follows from  the fact that $\bV$ is a function of $\bY$,~\eqref{eq:ach0} is due to the Markov chain $\Fak{\bY} - \bV - \bY$, and  \eqref{eq:ach1} follows from the faking procedure inducing $\Probc_{\Fak{\bY}|\bV}=\Probc_{\bY|\bV}$. Fano's inequality implies~\ref{eq:ach2} (assuming that $\cC\in\cG_1$). Note that the above bound is a multi-letter bound that depends on the specific codebook $\cC$. A single letter bound depending only on the probability distribution of the single letter random variables follows from concentration arguments over the codebook generation process.
In the following, we argue that, with high probability over the  generation of $\cC$, $\MI(\bX;\bY|\bV)\geq n\MI(X;Y|V)-n\epsilon$ for a large enough $n$. For every $\bv\in\cV^n$, let us define the multi-set $\cC_\bv\triangleq\{\bx\in\cC:(\bx,\bv)\in\AEN{X,V}\}$.\footnote{Recall that $\cC$ is a multi-set with possibly repeated elements. As a result, $\cC_\bv$ may also contain codewords that have multiplicity greater than one.} Further, for every $\bx\in\cX^n$, let $\cM_\bx\triangleq\{\bm\in\cM:\bx^{(\cC)}(\bm)=\bx\}$.
First, note that
\begin{align*}\lefteqn{\MI(\bX;\bY|\bV)\geq \ent(\bX|\bV)-n\epsilon}\\
\intertext{by Fano's inequality  (assuming that $\cC\in\cG_1$). Then,  given a code $\cC$, there exists $\epsilon'=\epsilon'(\epsilon)$ satisfying $\lim_{\epsilon\to 0}\epsilon'=0$ and}
\lefteqn{\ent(\bX|\bV)\geq\sum_{(\bx,\bv)\in\AEN{X,V}}\Probc_{\bX,\bV}(\bx,\bv)\log\frac{\Probc_\bV(\bv)}{\Probc_{\bX,\bV}(\bx,\bv)}}\\
&= \sum_{(\bx,\bv)\in\AEN{X,V}}\Probc_{\bX}(\bx)\Prob_{\bV|\bX}(\bv|\bx)\log\frac{\sum\limits_{\bx'\in\cC}2^{-nR}\ \Prob_{\bV|\bX}(\bv|\bx')}{\Probc_\bX(\bx)\Prob_{\bV|\bX}(\bv|\bx)}\\
&\geq \sum_{(\bx,\bv)\in\AEN{X,V}}\Probc_{\bX}(\bx)\Prob_{\bV|\bX}(\bv|\bx)\log\frac{\sum\limits_{\bx'\in\cC_\bv}2^{-nR}\ \Prob_{\bV|\bX}(\bv|\bx')}{\Probc_\bX(\bx)\Prob_{\bV|\bX}(\bv|\bx)}\\
&\overset{(a)}{\geq} \sum_{(\bx,\bv)\in\AEN{X,V}}\Probc_{\bX}(\bx)\Prob_{\bV|\bX}(\bv|\bx)\log\frac{\sum\limits_{\bx'\in\cC_\bv}\Prob_{\bV|\bX}(\bv|\bx')}{|\cM_\bx|\ \Prob_{\bV|\bX}(\bv|\bx)}\\
&\overset{(b)}{\geq} \sum_{(\bx,\bv)\in\AEN{X,V}}\Probc_{\bX}(\bx)\Prob_{\bV|\bX}(\bv|\bx)\log\frac{|\cC_\bv|}{|\cM_\bx|\ }-2n\epsilon'.\numberthis\label{eq:rxmultiletter}\\
\intertext{In the above, $(a)$ is obtained by expressing $\Probc_{\bX}(\bx)$ as $2^{-nR}|\cM_\bx|$. $(b)$ follows by noting that for every $(\bx,\bv)$ belonging to $\AEN{X,V}$, $\left|\log\frac{1}{\Prob_{\bV|\bX}(\bv|\bx)}-n\ent(V|X)\right|<n\epsilon'$ for some $\epsilon'>0$ that can be made arbitrarily close to $0$ as $\epsilon$ approaches $0$. We now show that, with high probability over the random generation of $\cC$, the expression in~\eqref{eq:rxmultiletter} is lower bounded in the desired manner. To this end, define the following three desirable events over the codebook generation process.}
\cG_2&\triangleq \left\{\cC:\sum_{(\bx,\bv)\in\AEN{X,V}}\Probc_{\bX}(\bx)\Prob_{\bV|\bX}(\bv|\bx)>(1-\epsilon)\right\}\\
\cG_3&\triangleq \left\{\cC:|\cC_\bv|\geq 2^{n(R-\MI(X;V)-\epsilon'')}\ \forall\ \bv\in\AEN{V}\right\}\\
\cG_4&\triangleq\left\{\cC:|\cM_\bx|<2^{n\epsilon}\ \forall\ \bx\in\AEN{X}\right\}.\\
\intertext{In the above, $\epsilon''>0$ is a constant that is specified later. Note that if $\cC\in\cap_{i=1}^4\cG_i$, then Eq.~\eqref{eq:ach2}-\eqref{eq:rxmultiletter} imply that}
D&\geq \frac{1}{n}(1-\epsilon)\log\frac{2^{n(R-\MI(X;V)-\epsilon'')}}{2^{n\epsilon}}-2(\epsilon+\epsilon')\\
&=\left((1-\epsilon)(\MI(X;Y|V)-\rho-\epsilon-\epsilon'')-2(\epsilon+\epsilon')\right)\\
&= \MI(X;Y|V)-(\rho+3\epsilon+\epsilon\MI(X;Y|V)+2\epsilon'+\epsilon'')\\
&\geq \MI(X;Y|V)-(\rho+3\epsilon+\epsilon\log|\cX|+2\epsilon'+\epsilon'').\numberthis\label{eq:rxdeniability}
\end{align*}
We next lower bound the probabilities of each of the above events. \\
{\noindent{\em i) Event $\cG_2$:}} First observe that \begin{align*}
\lefteqbr{\sum\limits_{(\bx,\bv)\in\AEN{X,V}}\Probc_{\bX}(\bx)\Prob_{\bV|\bX}(\bv|\bx)}{\geq \frac{|\cC\cap\AEN{X}|}{|\cC|}\min_{\bx\in\AEN{X}}\sum_{\bv\in\AEN{V|\bx}}\Prob_{\bV|\bX}(\bv|\bx)\label{eq:g2a}\numberthis.}
\intertext{To bound the right hand side above, we first note that using the additive form of the Chernoff bound and the definition of strong typicality,}
\lefteqbr{\expect_\cC\left[\frac{|\cC\cap\AEN{X}|}{|\cC|}\right]}{=\sum_{\bx\in\AEN{X}}\Prob_{\bX}(\bx)}\\
&\geq 1- |\cX| \max_{\tilde{x}\in\cX}\sum_{\bx:\ \left| \frac{1}{n}|\{i:x_i=\tilde{x}\}|-\Prob_X(\tilde{x})\right|>\frac{\epsilon}{|\cX|}}\Prob_\bX\left(\bx\right)\\
&\geq 1-|\cX|\exp{\left(-\frac{n\epsilon^2\min_{\tilde{x}\in\cX}\Prob_X(\tilde{x})}{4|\cX|^2}\right)}.\label{eq:chernoff1}\numberthis\\
\intertext{In particular, for a large enough $n$, we have}
\lefteqbr{\expect_\cC\left[\frac{|\cC\cap\AEN{X}|}{|\cC|}\right]}{\geq 1-\epsilon/4.}
\intertext{Next, by standard properties of the conditionally typical set, we have, for large enough $n$,}
%\lefteqbr{\Prob_{\bV|\bX}\left(\left.\bv\in\AEN{V|\bx}\right|\bx\right)}{\geq 1- |\cV||\cX|\exp{\left(-\frac{\min_{x\in\cX,v\in\cV}n\epsilon^2\Prob_{X,V}(x,v)}{4|\cV|^2}\right)}.\label{eq:chernoff2}\numberthis\\} 
\lefteqbr{\sum_{\bv\in\AEN{V|\bx}}\Prob_{\bV|\bX}(\bv|\bx)}{\geq 1- \epsilon/4.\label{eq:chernoff2}\numberthis}
\end{align*} 
Combining~\eqref{eq:chernoff1} and~\eqref{eq:chernoff2}, we conclude that there exists $n^*$ such that for every $n>n^*$, \begin{equation*}\expect_\cC\left[\frac{|\cC\cap\AEN{X}|}{|\cC|}\min_{\bx\in\AEN{X}}\sum_{\bv\in\AEN{V|\bx}}\Prob_{\bV|\bX}(\bv|\bx)\right]>1-\frac{\epsilon}{2}.\end{equation*}The above expression gives, in expectation, a lower bound on the left hand side of~\eqref{eq:g2a}. A further concentration argument over the i.i.d. generation of the codebook shows the existence of $n_2=n_2(\epsilon)$ such that whenever $n>n_2(\alpha)$,
\begin{align*}
{\Pr}_{\cC}(\cG_2)&\geq 1-\alpha/4.\label{eq:event1}\numberthis
\end{align*}\\
{\noindent\em ii) Event $\cG_3$:} Next, note that for any $\bv\in\AEN{V}$, there exists $n^{\#}$ and $\epsilon''=\epsilon''(\epsilon)$ satisfying $\lim_{\epsilon\to 0}\epsilon''=0$ and 
\begin{align*}
\expect\left|\cC_\bv\right|&=2^{nR}\sum_{\bx:(\bx,\bv)\in\AEN{X,V}}\Prob_\bX(\bx)\\
&\geq 2^{n(R-\MI(X;V)-\epsilon''/2)}\\
\intertext{whenever $n>n^{\#}$. Now, since each codeword falls in $\cC_\bv$ in an independent and identical manner over the codebook generation, the true value of $\cC_\bv$ concentrates around its mean with a high probability. In particular, by applying Chernoff bound on $\cC_\bv$, we obtain that there exists $n_3=n_3(\epsilon)$ such that for every $n>n_3(\alpha)$,}
{\Pr}_{\cC}(\cG_3)&\geq{\Pr}_{\cC}(|\cC_\bv|\geq 2^{-n\epsilon''/2}\expect|\cC_\bv|)\\
& >1-\alpha/4.\label{eq:event2}\numberthis
\end{align*}
{\noindent\em iii) Event $\cG_4$:} Finally,  let $\beta=2^{n\epsilon}$, and observe that there exists $\epsilon'''=\epsilon'''(\epsilon)$ such that $\lim_{\epsilon\to 0}\epsilon'''(\epsilon)=0$ and
\begin{align*}
\lefteqbr{\log \left({\Pr}_{\cC}(\cC\notin\cG_4)\right)}{=\log{\Pr}_{\cC}{\Large(}\exists\ \cS\subseteq \cM, \bx\in\AEN{X}\mbox{ s.t. }|\cS|= \beta\eqbr\mbox{ and }\bx^{(\cC)}(\bm)=\bx\ \forall\ \bm\in\cS{\Large)}}\\  
&\leq \log\sum_{\substack{\cS\subseteq\cM\\ |\cS|=\beta}} \sum_{\bx\in\AEN{X}}\prod_{\bm\in\cS}{\Pr}_{\cC}\left(\bx^{(\cC)}(\bm)=\bx\right)\\
&= \log{|\cM|\choose\beta}+\log\sum_{\bx\in\AEN{X}}\left({\Pr}_{\cC}\left(\bx^{(\cC)}(1)=\bx\right)\right)^\beta\\ 
&\overset{(a)}{\leq} |\cM|H_b\left(\frac{\beta}{|\cM|}\right)+\log|\AEN{X}| \eqbr +\beta\log\max_{\bx\in\AEN{X}}{\Pr}_{\cC}(\bx^{(\cC)}(1)=\bx)\\
&\overset{(b)}{\leq} \beta\log\frac{|\cM|}{\beta}+(|\cM|-\beta)\log\frac{|\cM|}{|\cM|-\beta}-(\beta-1) n\ent(X) \eqbr +(\beta+1)n\epsilon'''\\
&\overset{(c)}{\leq}\beta\log\frac{|\cM|}{\beta}+\beta\log e -(\beta-1) n \ent(X)+(\beta+1)n\epsilon'''\\
&= 2^{n\epsilon}(n(R-\epsilon)+\log e)-(2^{n\epsilon}-1)n\ent(X)+(2^{n\epsilon+1})n\epsilon'''\\
&= 2^{n\epsilon}(n(R-\ent(X)-\epsilon+\epsilon''')+\log e)+n(\ent(X)+\epsilon''')\\
&\leq 2^{n\epsilon}(n(R-\MI(X;Y)-\epsilon+\epsilon''')+\log e)+n(\ent(X)+\epsilon''')\\
& = 2^{n\epsilon}(n(-\rho-\epsilon+\epsilon''')+\log e)+n(\ent(X)+\epsilon''')\label{eq:event3log}.\numberthis
\intertext{In the above, $(a)$ is a standard upper bound on ${|\cM|\choose\beta}$ in terms of the binary entropy function $H_b(\beta/|\cM|)$. $(b)$ obtained by noting that there exists $\epsilon'''$ such that $\lim_{\epsilon\to 0}\epsilon'''=0$, $|\AEN{X}|\leq2^{n(\ent(X)+\epsilon''')}$ and $\Prob_{\bX}(\bx)\leq 2^{-n(\ent(X)-\epsilon''')}$ for each $\bx\in\AEN{X}$. Lastly, $(c)$ is obtained by using the fact that for every $a>0$, $\log a=\log e\ln a\leq (a-1)\log e$. Note that as long as $\rho$  is strictly greater than $\epsilon'''-\epsilon$, the right hand side of~\eqref{eq:event3log} diverges to $-\infty$ as $n$ increases without bound. In particular, this implies that there exists $n_4$ such that for every $n>n_4(\alpha)$,}
&{\Pr}_{\cC}(\cG_4)>1-\alpha/4.\label{eq:event3}\numberthis
\end{align*}
Finally, combining~\eqref{eq:rxreliability},~\eqref{eq:event1},~\eqref{eq:event2}, and~\eqref{eq:event3} we conclude that, whenever $n>\max\{n_1,n_2,n_3,n_4\}$, with probability at least $1-\alpha$, the randomly drawn code is simultaneously $(\epsilon,R)$-reliable and $(0,D)$-plausibly deniable where $(R,D)$ satisfy the lower bounds in~\eqref{eq:rxrate} and~\eqref{eq:rxdeniability}. Since $\rho$ and $\epsilon$ can be made arbitrarily close to zero, this shows the achievability of all rates in the interior of the claimed region.\end{proof}

\subsection{Discussions}
\subsubsection{An example} 
\begin{example} \label{ex:TX} Consider a channel $\Prob_{YZ|X}$ with $\cX=\cY=\cZ=\set{1,2,3}$, $Y=X$ and $\Prob_{Z|X}$ as in Figure~\ref{fig:exampleziv}, {\em i.e.}, $$\Prob_{YZ|X}(y,z|x)=\left\{\begin{array}{ll} 0.3 & (x,y,z)\in\set{(1,1,1),(2,2,1)}\\ 0.7& (x,y,z)\in\set{(1,1,2),(2,2,2)}\\ 0.4& (x,y,z)=(3,3,2)\\0.6& (x,y,z)=(3,3,3)\\0 &\mbox{otherwise.}\end{array}\right.$$
	We characterize the capacity region $\cR_{\bx}$ by restricting our choice of the auxililary random variable $U$ to the zero-information random variable. For the above conditional distribution, the zero-information random variable of $X$ w.r.t. $\Prob_{Z|X}$ takes two values: $u_{1}=\set{1,2}$ and $u_{2}=\set{3}$. Since $X=Y$, $\MI(X;Y)=\ent(X)$ and $\MI(X;Y|U)=\Prob_{X}(1)\log\frac{\Prob_{X}(\set{1,2})}{\Prob_X(1)}-\Prob_X(2)\log\frac{\Prob_{X}(\set{1,2})}{\Prob_X(2)}$. The capacity region $\cR_{\bx}$ (Figure~\ref{fig:TX}) consists of all $(R,D)$ pairs satisfying the following
	\begin{align*}
	D&\leq R\leq H_t\left(\frac{D}{2},\frac{D}{2},1-D\right)\\	
	0&\leq D\leq 1,
	\end{align*}
where $H_t(\cdot,\cdot,\cdot)$ represents the ternary entropy function. Interestingly, the capacity region depends on the conditional distribution $\Prob_{Z|X}$,  only through the zero-information variable induced by it -- all conditional distributions $\Prob_{Z|X}$ that induce the same zero-information variable have the same capacity region (assuming $\Prob_{Y|X}$ is unchanged). This is a general feature of capacity regions for the transmitter deniability problem. 
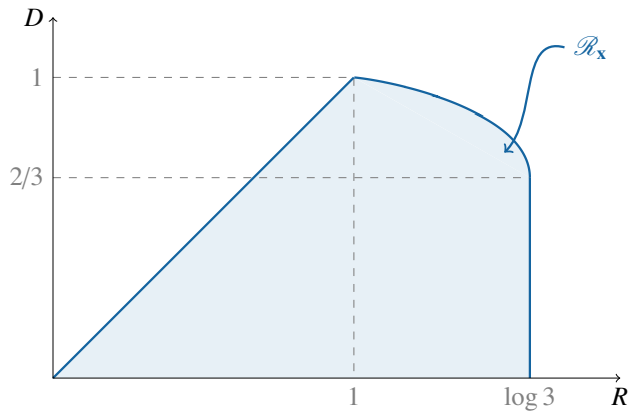
\begin{figure}[h]
\begin{center}
\begin{tikzpicture}[baseline, scale=4]

%\begin{axis}[enlargelimits=0.1]
%\addplot[domain=-.15:1,fill=gray!50] {x^2}\closedcycle;
%\end{axis}           
\fill[myblue!10] plot [domain={2/3}:1,samples=1000] ({\x-\x*log2(\x)-(1-\x)*log2(1-\x)},\x);
\fill[myblue!10] (0,0) -- (1,1) -- ({log2(3)},{2/3}) -- ({log2(3)},0) -- cycle;  
\draw[gray, dashed] (0,1) node[left]{$1$} -- (1,1) -- (1,0) node[below]{$1$};
\draw[gray, dashed] (0,2/3) node[left]{$2/3$} -- ({log2(3)},2/3) -- ({log2(3)},0) node[below]{$\log{3}$};
\draw[line width=.3mm,myblue] plot [ domain=0:1,samples=1000] (\x,\x) ;
\draw[line width=.3mm,myblue] plot [domain={2/3}:1,samples=1000] ({\x-\x*log2(\x)-(1-\x)*log2(1-\x)},\x);
\draw[line width=.3mm,myblue] plot [domain={2/3}:0,samples=1000] ({log2(3)},\x);
\draw[line width =0.3mm, <-,myblue] (1.5,.75) to[out=45,in=165]   (1.7,1.1) node[right] {$\cR_\bx$};
\draw[<->, line width=0.1mm] ({log2(3)+0.3},0) node[below]{$R$} -- (0,0) -- (0,1.2) node[left]{$D$};
\end{tikzpicture}\end{center}	
\caption{Capacity region $\cR_\bx$ for Example~\ref{ex:TX}.}\label{fig:TX}
\end{figure}
\end{example}

\subsubsection{Rate of deniability as the Equivocation rate} Similar to the Message Deniability setting, we can attach a secrecy interpretation to the rate of deniability for faking procedures that are plausibly deniable. The following proposition mirrors Proposition~\ref{prop:msgDEQ}.
\begin{proposition}\label{prop:cweq} Let $\Fak{\bX}$  be $(\delta,D)$-plausibly deniable for $\bX$ given $\bZ$ and satisfy the Markov chain $\Fak{\bX}-\bX-\bZ$. Then, there exists $\mu\geq 0$ depending only on $\Prob_{Z|X}$ such that
$$nD-n\mu\sqrt{\delta}\leq\ent(\bM|\bZ,\Fak{\bX})\leq nD+n\mu\sqrt{\delta}.$$
\end{proposition}
\begin{proof} The proof relies on the Lemma~\ref{lem:xzgivenf} and proceeds in similar spirit as Proposition~\ref{prop:msgDEQ}. To this end, let $\kappa$ be the constant defined in Lemma~\ref{lem:xzgivenf}. Note that
\begin{align*}
	\ent(\bM|\bZ,\Fak{\bX})&=\ent(\bM|\Fak{\bX})-\MI(\bM;\bZ|\Fak{\bX})\\
	&= nD + \ent(\bM|\Fak{\bX})-\ent(\Msg(\Fak{\bX})|\bX)\eqbr -\MI(\bM;\bZ|\Fak{\bX}).
\end{align*}
Applying the non-negativity of mutual information and Lemma~\ref{lem:xzgivenf}, we obtain 
\begin{align*}
	\ent(\bM|\bZ,\Fak{\bX})&\leq nD+\ent(\bM|\Fak{\bX})-\ent(\Msg(\Fak{\bX})|\bX) \\
	&\leq nD+n\kappa\sqrt{\delta},
\intertext{and}
\ent(\bM|\bZ,\Fak{\bX})&\geq nD-2n\kappa\sqrt{\delta}.
\end{align*}
Choosing $\mu=2\kappa$ completes the proof.\end{proof}

\section{Concluding remarks}\label{sec:discussions}
In this paper, we have considered three different models of Plausible Deniability and give achievable rates for each model while also giving tight converses for the message deniability and transmitter deniability settings.  It is evident that, at the very least, each capacity region is a subset of the {\em Rate-Equivocation} region. Intuitively, this may be interpreted as follows --  any code that has a rate of deniability $D$ has the property that the equivocation at the eavesdropper is at least $D$ (otherwise, with high probability, the eavesdropper can detect a fake response). On the other hand, it is not {\em a priori} clear whether the achievable rates for any one model considered in this paper is a subset of another -- part of the difficulty in comparing the different settings arises from the fact that in each setting, the faking procedure accepts different inputs to generate the fake output.

Digging deeper into the nature of our problem, our achievability proofs rely crucially on the summoned party's ability to identify a set of \emph{plausible} fake responses that appear roughly as likely as the true response to an eavesdropper who also observes the channel output. Further, the set of plausible responses must be identified without knowing the eavesdropper actual channel observation. To achieve this goal, our schemes ensure that the set of possible response values partitions into ``cliques'' such that each response from the clique would be plausible to the eavesdropper given any likely channel output. This simplifies our faking procedure to randomly picking one response from the clique corresponding to the true response. In our scheme for the message deniability setting, these cliques correspond to all messages that are consistent with the transmitted ``public message'',  while in the transmitter and receiver deniability settings, these cliques correspond to codewords and received vectors that are statistically consistent with the zero information variables of the actual transmitted codeword and the received vector, respectively. In each of these settings, given the clique corresponding to the true value of the summoned party's response, the eavesdropper's channel observation provides  asymptotically negligible additional information about the true value of the response.

Given our problem formulation, the above achievability idea appears natural. Perhaps surprisingly, we also show that any good faking procedure for our problem must follow the above decomposition (at least roughly). In the transmitter and receiver deniability settings, Lemmas~\ref{lem:convmessage} and~\ref{lem:xzgivenf} make this claim precise. A drastic consequence of this is that non-zero rates are possible for transmitter deniability only when non-trivial zero information variables exist with respect to the eavesdropper's channel output. We note that the existence of such variables is guaranteed only for fairly special classes of channels -- even for channels such as Binary Symmetric Channels, the only zero information variables are the channel inputs themselves. Further, the existence of non-trivial zero information variables may be rather fragile with respect to perturbations in the channel conditional probability. This is in contrast to the message deniability setting, the capacity region for plausible deniability seems somewhat robust to the channel statistics (\emph{c.f.}~\cite{SchaeferB:14} for the robustness analysis for a related problem). 

Our work potentially leads to several intriguing open questions. In settings where non-zero rates of deniability are not possible (\emph{e.g.} transmitter deniability over a binary symmetric broadcast channel), it is of interest to understand whether an asymptotically vanishing rate of communication may still be possible. Recent work on ``{\em square-root law}'' in covert communications~\cite{BashGT:13,CheBJ:13,WangWZ:16,Bloch:16} suggests such a possibility. However, unlike covert communication, the eavesdropper in our setting has potentially greater distinguishing power due to access to both the channel observation and the summoned party's response. Separately, while our work examines the broadcast channel setting, the notion of information theoretic plausible deniability readily extends to other communication settings with security oriented goals, \emph{e.g.}, secret key generation, interactive communication, and communication with public discussion. It would be interesting to examine the capacity question in these settings. Finally, we remark that while our formulation of plausible deniability relies on the asymmetry between the channel to the eavesdropper and the legitimate receiver, and the cryptographic formulation of~\cite{CanettiDNO:97,ONeillPW:11,SahaiW:14} relies on the eavesdropper's inability to  efficiently compute certain functions without knowing the receiver's private key, it would be interesting to understand whether other forms of asymmetry between the legitimate parties and the eavesdropper may be similarly exploited to obtain plausibly deniable communication.

%\printbibliography

\appendices
\section{Strong secrecy for Broadcast Channels with both Confidential and Leaked Messages}\label{app:strongsecrecysideinfo}
In the following, we consider the problem of broadcast channel with both confidential and leaked messages described in Figure~\ref{fig:secrecy}. We first give the proof of Lemma~\ref{lem:sideinfo} that gives an inner bound on the capacity region defined in Definition~\ref{def:bccsi}.
\begin{proof}[Proof of Lemma~\ref{lem:sideinfo}]
The proof essentially follows from the strategies used in ~\cite[Theorem 17.13]{CsiszarK:11},~\cite[Theorem 3]{BlochL:13},~\cite{MatH:10} to prove strong secrecy capacity region for the setting of {\em broadcast channel with confidential messages} (Figure~\ref{fig:WTCBCC}). The only difference here from the settings of~\cite{CsiszarK:11,BlochL:13,MatH:10},  is that we do not demand that the message $\bt$ be reliably decoded by Judy from her observation $\bz$. This allows us to send $\bt$ at all rates less than $\MI(V;Y)$ instead of $\min \set{\MI (V;Y), \MI(V;Z)}$ for every $(U,V)$ pair satisfying the lemma conditions. As the proof would be nearly identical to the proofs supplied in~\cite{CsiszarK:11,BlochL:13,MatH:10}, we skip the full proof of the lemma here. \end{proof}

Next, we give a lemma that allows us to modify the strong secrecy metric (condition 3 of Definition~\ref{def:bccsi}) to the Kullback-Leibler Divergence in the form suitable for our problem. 
\begin{lemma}\label{lem:reversekl} Let $\alpha\in(0,1)$ and $\beta>0$. Let $\tilde{I},J$ be random variables distributed on $\cal I$ and $\cal J$ respectively with a joint distribution $\Prob_{\tilde{I},J}$ such that $\MI(\tilde{I};J)<\beta$.  Then, there exists a random variable $I\in\cal I$ that is jointly distributed with $\tilde{I}$ and $J$ in accordance with a Markov chain $I-\tilde{I}-J$ such that the joint distribution $\Prob_{I,\tilde{I},J}$ has the following properties:  
\begin{enumerate}
\item\label{cond:error} $\Prob_{I,\tilde{I}}\left(I\neq \tilde{I}\right)>1-\alpha$.
\item\label{cond:identical} $\Prob_I(i)=\Prob_{\tilde{I}}(i)$ for all $i\in\cal I$.
\item\label{cond:mutualinfo} $\MI(I;J)<\beta$.
\item \label{cond:reverseKL} $\KL(\Prob_{I}\Prob_J||\Prob_{I,J})\leq \sqrt{2\beta}\log\left(1/\alpha\right)$.

%\item\label{cond:mutualinfo} $\Prob_{I,\tilde{I},J}(i,\tilde{i},j)=\Prob_{\tilde{I}}(i)\Prob_{\tilde{I}|I}(\tilde{i}|i)\Prob_{J|\tilde{I}}(j|\tilde{i})$ for all $(i,\tilde{i},j)\in I\times I\times J$.
\end{enumerate}
 \end{lemma}
%\begin{remark} If $\KL(\Prob_{I}\Prob_J||\Prob_{I,J})$ in requirement~\ref{cond:reverseKL}  is replaced by $\KL(\Prob_{I,J}||\Prob_{I}\Prob_J)$, the lemma is trivially true as the latter quantity equals the mutual information $\MI(I;J)$.	
%\end{remark}
\begin{proof} In the following, we assume, without loss of generality, that $\Prob_{\tilde{I}}(\tilde{i})>0$ and $\Prob_J(j)>0$ for each $(\tilde{i},j)\in\mathcal{I}\times\mathcal{J}$. We construct the random variable $I$ explicitly as follows. First,  we define the transition probability 
$$\Prob_{I|\tilde{I}}(i|\tilde{i})=\left\{\begin{array}{ll} 1-\alpha +\alpha \Prob_{\tilde{I}}(\tilde{i})&i=\tilde{i}\\ \alpha\Prob_{\tilde{I}}(i) &i\neq\tilde{i},\end{array}\right.$$
and let  $\Prob_{I,\tilde{I},J}(i,\tilde{i},j)=\Prob_{I|\tilde{I}}(i|\tilde{i})\Prob_{\tilde{I},J}(\tilde{i},j)$ for all $(i,\tilde{i},j)\in I\times I\times J$. Clearly, $I$ equals $\tilde{I}$ with probability at least $1-\alpha$. Hence, condition~\ref{cond:error} is satisfied. Also, $\Prob_{I}(i)=\sum_{\tilde{i}\in\cal I}	\Prob_{I|\tilde{I}}(i|\tilde{i})\Prob_{\tilde{I}}(\tilde{i})=\Prob_{\tilde{I}}(i)$, which implies that condition~\ref{cond:identical} is also satisfied. Further, noting that $I-\tilde{I}-J$ is a Markov chain, by the Data Processing inequality, 
\begin{align*}
\MI(I;J)&\leq \MI	(\tilde{I};J)\\
&<\beta.
\end{align*}
Thus, condition~\ref{cond:mutualinfo} is satisfied as well. Note that 
\begin{align*}
\Prob_{I,J}(i,j)&=\sum_{\tilde{i}\in\mathcal{I}}\Prob_{I|\tilde{I}}(i|\tilde{i})\Prob_{\tilde{I},J}(\tilde{i},j)\\
&= (1-\alpha+\alpha \Prob_I(i))\Prob_{\tilde{I},J}(i,j)+ \alpha\Prob_I(i)\sum_{\tilde{i}\in\mathcal{I}\setminus\set{i}}\Prob_{\tilde{I},J}(\tilde{i},j)\\
&\geq \alpha \Prob_I(i)\Prob_J(j).\numberthis\label{eq:pij}
\end{align*}
Note that, by our assumption, $\Prob_{I}(i)\Prob_{J}(j)>0$ for each $(i,j)\in\mathcal{I}\times\mathcal{J}$. Further, Eq.~\eqref{eq:pij} implies that $\Prob_{I,J}(i,j)>0$ for each $(i,j)\in\mathcal{I}\times\mathcal{J}$. Thus, $\KL(\Prob_I\Prob_J||\Prob_{I,J})$ and $\KL(\Prob_{I,J}||\Prob_I\Prob_J)$ are finite and well-defined. Now,
\begin{align}
\lefteqn{\KL(\Prob_I\Prob_J||\Prob_{I,J})= \sum_{i\in \mathcal{I},j\in\mathcal{J}}\Prob_I(i)\Prob_{J}(j)\log\frac{\Prob_I(i)\Prob_{J}(j)}{\Prob_{I,J}(i,j)}}\\
	& =  \sum_{i\in\mathcal{I},j\in\mathcal{J}}\left(\Prob_I(i)\Prob_{J}(j)-\Prob_{I,J}(i,j)\right)\log\frac{\Prob_I(i)\Prob_{J}(j)}{\Prob_{I,J}(i,j)}\eqbr +\sum_{i\in\mathcal{I},j\in\mathcal{J}}\Prob_{I,J}(i,j)\log\frac{\Prob_I(i)\Prob_{J}(j)}{\Prob_{I,J}(i,j)}\\
&\overset{(a)}{\leq} ||\Prob_{I,J}-\Prob_{I}\Prob_J||_1\max_{i\in\mathcal{I},j\in\mathcal{J}}\log\frac{\Prob_I(i)\Prob_{J}(j)}{\Prob_{I,J}(i,j)}-\KL(\Prob_{IJ}||\Prob_I\Prob_J)\\
& \overset{(b)}{\leq }\sqrt{2\beta}\log(1/\alpha).
\end{align}
In the above, $(a)$ follows from H\"older's inequality. $(b)$ is obtained by applying the non-negativity of the Kullback-Leibler Divergence, inequality~\eqref{eq:pij}, and by noting that 
\begin{align}
||\Prob_{I,J}-\Prob_I\Prob_J||_1&=\sum_{i\in\mathcal{I},j\in\mathcal{J}}\left|\sum_{\tilde{i}\in\mathcal{I}}\Prob_{I|\tilde{I}}(i|\tilde{i})\left(\Prob_{\tilde{I},J}(\tilde{i},j)-\Prob_{\tilde{I}}(\tilde{i})\Prob_J(j)\right)\right|\\
&\leq\sum_{i\in\mathcal{I},\tilde{i}\in\mathcal{I}j\in\mathcal{J}}\Prob_{I|\tilde{I}}(i|\tilde{i})\left|\Prob_{\tilde{I},J}(\tilde{i},j)-\Prob_{\tilde{I}}(\tilde{i})\Prob_J(j)\right|\\
&=||\Prob_{\tilde{I},J}-\Prob_{\tilde{I}}\Prob_J ||_1 \\
&\overset{(a)}{\leq} \sqrt{2\KL\left(\Prob_{\tilde{I},J}||\Prob_{\tilde{I}}\Prob_J\right)}\\
&\leq  \sqrt{2\beta}.\end{align}
In the above, $(a)$ follows from Pinsker's inequality. This proves that $I$ satisfies the conditions~\ref{cond:identical}-\ref{cond:reverseKL}.\end{proof}

Finally, we give a proof of Corollary~\ref{cor:sideinfo}. 
\begin{proof}[Proof of Corollary~\ref{cor:sideinfo}] The proof follows by starting with a code from Lemma~\ref{lem:sideinfo} and using Lemma~\ref{lem:reversekl} to modify it to achieve the desired properties. Let $\tilde{\cC}$ be a code that satisfies conditions 1-3 of Definition~\ref{def:bccsi}. Let $\tilde{\bS}\in\cS$ denote the confidential message and $\bT\in\cT$ denote the leaked message for this code. Note that the random variables $\tilde{\bS},\bT$, and $\bZ$ satisfy $\MI(\tilde{\bS};\bT,\bZ)<\delta$. Next, apply Lemma~\ref{lem:reversekl} with $\tilde{\bS}$, $(\bT,\bZ)$, $\epsilon$, and $\delta$, in place of $\tilde{I}$, $J$, $\alpha$, and $\beta$, respectively to obtain the random variable $\bS$ (in place of $I$) that is jointly distributed with $\tilde{\bS}$ and $(\bT,\bZ)$ according to a distribution $\Probc_{\bS,\tilde{\bS},(\bT,\bZ)}=\Probc_{\bS|\tilde{\bS}}\Probc_{\tilde{\bS}}\Probc_{(\bT,\bZ)|\tilde{\bS}}$. 

Consider a code $\cC$ that operates as follows. Let $(\bS,\bT)$ be the messages for this code. First, Alice maps the message $\bS$ to a randomly drawn $\tilde{\bS}$ according to the transition probability $\Probc_{\tilde{\bS}|\bS}$.  Next, she encodes $(\tilde{\bS},\bT)$ using the encoder for $\tilde{\cC}$. Upon receving $\bY$, Bob uses the decoder for $\tilde{\cC}$ to output his reconstruction of $(\bS,\bT)$.

By Lemma~\ref{lem:reversekl}, the overall code  satisfies the conditions of Definition~\ref{def:bccsi} with requirement 2 replaced by $$\sum_{(\by,\bs,\bt):{\small \Dec}(\by)\neq(\bs,\bt)}\Probc_{\bY,\bS,\bT}(\by,\bs,\bt)\leq 2\epsilon.$$ In addition, the code also satisfies the following  property  $$\KL(\Probc_{\bS}\Probc_{\bT,\bZ}||\Probc_{\bS,\bT,\bZ})<\sqrt{2\delta}\log(1/\epsilon).$$  
Now, by first choosing $\epsilon$ small enough and subsequently, $\delta$ small enough, both the error probability and the K-L divergence above can be made arbitrarily small. This proves the corollary.
\end{proof}

\section{A continuity property}
\begin{lemma}\label{lem:continuity}
Let $\mathscr{P}$ be a compact subset of the set of probability measures over a finite set $\cB$. Let $\EL:\mathscr{P}\to\bbr^{+}$ and $\EM:\mathscr{P}\to\bbr^+$ be functionals that are continuous with respect to the variational distance such that $\EL^{-1}(\{0\})\neq\phi$ . Then,
\begin{equation}\lim_{\delta\to 0^+}\max_{\Prob\in\mathscr{P}:\EL(\Prob)<\delta}\EM(\Prob)=\max_{\Prob\in\mathscr{P}:\EL(\Prob)=0}\EM(\Prob).\label{eq:continuitylemma}\end{equation}
\end{lemma}
\begin{proof}
Since $\mathscr{P}$ is compact and $\EM$ is a continuous on $\mathscr{P}$, $\EM$ is bounded. Further, as $\EM(\Prob)\geq 0$ for every $\Prob\in\mathscr{P}$, and $\max_{\Prob\in\mathscr{P}:\EL(\Prob)<\delta}\EM(\Prob)$ is an increasing function of $\delta$, the limit on the left hand side of Eq~\eqref{eq:continuitylemma} exists. Now, for any $\delta>0$,
\begin{align}
\max_{\Prob\in\mathscr{P}:\EL(\Prob)<\delta}\EM(\Prob)&\geq \max_{\Prob\in\mathscr{P}:\EL(\Prob)=0}\EM(\Prob).
\end{align}
 Taking the limit as $\delta$ approaches zero, the left hand side of Eq.~\eqref{eq:continuitylemma} is at least as large as the right hand side. Next, we show that the limit on the left hand side cannot be larger than the right hand side. 

To this end, let $M^*=\lim_{\delta\to 0}\max_{\Prob\in\mathscr{P}:\EL(\Prob)<\delta}\EM(\Prob)$. Thus, there exists a sequence  $\{\Prob^{(i)}\}_{i\in\bbn}$ in $\mathscr{P}$ such that $\EL(\Prob^{(i)})<1/i$ and $\lim_{i\to\infty}\EM(\Prob^{(i)})=M^*$. As $\mathscr{P}$ is a compact set under the variational distance, $\{\Prob^{(i)}\}_{i\in\bbn}$ contains a subsequence $\{\Prob^{(i_j)}\}_{j\in\bbn}$  that converges (in variational distance) to a limiting distribution $\Prob^*$. By continuity of $\EL$, we have $$
0\leq \EL(\Prob_\bB^*)=\lim_{j\to\infty}\EL(\Prob^{(i_j)})\leq\lim_{j\to\infty} 1/i_j=0.$$
Thus, $M^*=\EM(\Prob^*)\leq\max_{\set{\Prob\in \mathscr{P}:\EL(\Prob)=0}}\EM(\Prob)$.\end{proof}

\bibliographystyle{IEEEtran}
\bibliography{references}

\begin{thebibliography}{10}
\providecommand{\url}[1]{#1}
\csname url@samestyle\endcsname
\providecommand{\newblock}{\relax}
\providecommand{\bibinfo}[2]{#2}
\providecommand{\BIBentrySTDinterwordspacing}{\spaceskip=0pt\relax}
\providecommand{\BIBentryALTinterwordstretchfactor}{4}
\providecommand{\BIBentryALTinterwordspacing}{\spaceskip=\fontdimen2\font plus
\BIBentryALTinterwordstretchfactor\fontdimen3\font minus
  \fontdimen4\font\relax}
\providecommand{\BIBforeignlanguage}[2]{{%
\expandafter\ifx\csname l@#1\endcsname\relax
\typeout{** WARNING: IEEEtran.bst: No hyphenation pattern has been}%
\typeout{** loaded for the language `#1'. Using the pattern for}%
\typeout{** the default language instead.}%
\else
\language=\csname l@#1\endcsname
\fi
#2}}
\providecommand{\BIBdecl}{\relax}
\BIBdecl

\bibitem{KatzL:07}
J.~Katz and Y.~Lindell, \emph{Introduction to {M}odern {C}ryptography}.\hskip
  1em plus 0.5em minus 0.4em\relax Chapman \& Hall/CRC, 2007.

\bibitem{Wyner:75}
A.~Wyner, ``{The Wire-tap Channel},'' \emph{Bell System Technical Journal,
  The}, vol.~54, no.~8, pp. 1355--1387, Oct 1975.

\bibitem{CsiszarK:78}
I.~Csiszar and J.~K{\"o}rner, ``Broadcast {C}hannels with {C}onfidential
  {M}essages,'' \emph{IEEE Transactions on Information Theory}, vol.~24, no.~3,
  pp. 339--348, May 1978.

\bibitem{BlochB:11}
M.~Bloch and J.~Barros, \emph{{Physical-Layer Security: From Information Theory
  to Security Engineering}}.\hskip 1em plus 0.5em minus 0.4em\relax Cambridge
  University Press, 2011.

\bibitem{CaiY:02}
N.~Cai and R.~W. Yeung, ``Secure network coding,'' in \emph{Information Theory,
  2002. Proceedings. 2002 IEEE International Symposium on}, 2002, pp. 323--.

\bibitem{CanettiDNO:97}
R.~Canetti, C.~Dwork, M.~Naor, and R.~Ostrovsky, ``Deniable {E}ncryption,'' in
  \emph{Proceedings of the 17th Annual International Cryptology Conference on
  Advances in Cryptology}.\hskip 1em plus 0.5em minus 0.4em\relax London, UK:
  Springer-Verlag, 1997, pp. 90--104.

\bibitem{BenT:94}
J.~Benaloh and D.~Tuinstra, ``Uncoercible communication,'' \emph{Computer
  Science Technical Report TR-MCS-94-1, Clarkson University}, 1994.

\bibitem{ONeillPW:11}
\BIBentryALTinterwordspacing
A.~O'Neill, C.~Peikert, and B.~Waters, \emph{Bi-Deniable Public-Key
  Encryption}.\hskip 1em plus 0.5em minus 0.4em\relax Berlin, Heidelberg:
  Springer Berlin Heidelberg, 2011, pp. 525--542.
  \url{http://dx.doi.org/10.1007/978-3-642-22792-9_30}
\BIBentrySTDinterwordspacing

\bibitem{Truecrypt}
``Truecrypt: Hidden operating system,'' \url{http://truecrypt.sourceforge.net}.

\bibitem{SahaiW:14}
A.~Sahai and B.~Waters, ``How to {U}se {I}ndistinguishability {O}bfuscation:
  {D}eniable {E}ncryption, and {M}ore,'' in \emph{Proceedings of the 46th
  Annual ACM Symposium on Theory of Computing}.\hskip 1em plus 0.5em minus
  0.4em\relax New York, NY, USA: ACM, 2014, pp. 475--484.

\bibitem{BashGT:13}
B.~Bash, D.~Goeckel, and D.~Towsley, ``Limits of {R}eliable {C}ommunication
  with {L}ow {P}robability of {D}etection on {AWGN C}hannels,'' \emph{IEEE
  Journal on Selected Areas in Communications}, vol.~31, no.~9, pp. 1921--1930,
  September 2013.

\bibitem{CheBJ:13}
P.~H. Che, M.~Bakshi, and S.~Jaggi, ``{Reliable Deniable Communication: Hiding
  Messages in Noise},'' in \emph{Proceedings of the 2013 IEEE International
  Symposium on Information Theory}, July 2013, pp. 2945--2949.

\bibitem{Bloch:16}
M.~R. Bloch, ``Covert communication over noisy channels: A resolvability
  perspective,'' \emph{IEEE Transactions on Information Theory}, vol.~62,
  no.~5, pp. 2334--2354, May 2016.

\bibitem{WangWZ:16}
L.~Wang, G.~W. Wornell, and L.~Zheng, ``Fundamental limits of communication
  with low probability of detection,'' \emph{IEEE Transactions on Information
  Theory}, vol.~62, no.~6, pp. 3493--3503, June 2016.

\bibitem{CheBCJ:14}
P.~H. Che, M.~Bakshi, C.~Chan, and S.~Jaggi, ``Reliable, deniable and hidable
  communication,'' in \emph{2014 Information Theory and Applications Workshop
  (ITA)}, Feb 2014, pp. 1--10.

\bibitem{ElGamalK:11}
A.~El~Gamal and Y.-H. Kim, \emph{Network {I}nformation {T}heory}.\hskip 1em
  plus 0.5em minus 0.4em\relax Cambridge University Press, 2011.

\bibitem{AhlswedeK:75}
R.~Ahlswede and J.~Korner, ``Source coding with side information and a converse
  for degraded broadcast channels,'' \emph{IEEE Transactions on Information
  Theory}, vol.~21, no.~6, pp. 629--637, November 1975.

\bibitem{SchaeferB:14}
R.~F. Schaefer and H.~Boche, ``Robust broadcasting of common and confidential
  messages over compound channels: Strong secrecy and decoding performance,''
  \emph{IEEE Transactions on Information Forensics and Security}, vol.~9,
  no.~10, pp. 1720--1732, Oct 2014.

\bibitem{CsiszarK:11}
I.~Csiszar and J.~K{\"o}rner, \emph{{Information Theory: Coding Theorems for
  Discrete Memoryless Systems}}.\hskip 1em plus 0.5em minus 0.4em\relax
  Cambridge University Press, 2011.

\bibitem{BlochL:13}
M.~R. Bloch and J.~N. Laneman, ``Strong secrecy from channel resolvability,''
  \emph{IEEE Transactions on Information Theory}, vol.~59, no.~12, pp.
  8077--8098, Dec 2013.

\bibitem{MatH:10}
\BIBentryALTinterwordspacing
R.~Matsumoto and M.~Hayashi, ``Strong security and separated code constructions
  for the broadcast channels with confidential messages,'' \emph{CoRR}, vol.
  abs/1010.0743, 2010.  \url{http://arxiv.org/abs/1010.0743}
\BIBentrySTDinterwordspacing

\end{thebibliography}

\end{document}